\documentclass[11pt, a4paper]{article}
\usepackage[letterpaper, left=1in, right=1in, top=0.9in, bottom=0.9in]{geometry}

\AtBeginDocument{%
  }

\usepackage{amsmath}
\usepackage{amssymb}
\usepackage{amsfonts}
\usepackage{amsthm} 
\usepackage{graphicx}
\usepackage{comment}

\usepackage{hyperref}
\usepackage{here}

\usepackage{cite}

\newtheorem{theorem}{Theorem}[section]
\newtheorem{lemma}[theorem]{Lemma}

\newtheorem{claim}[theorem]{Claim}

\newtheorem{definition}[theorem]{Definition}
\newtheorem{remark}[theorem]{Remark}

\global\long\def\M{\mathcal{M}}%
\global\long\def\I{\mathcal{I}}%
\global\long\def\Ss{S_1, \dots, S_k}%
\global\long\def\tO{\tilde{O}}%

\usepackage[linesnumbered,noend,ruled,vlined]{algorithm2e}

\SetKwInput{KwInput}{Input}     
\SetKwInput{KwOutput}{Output}   
\SetKw{KwTo}{to}

\usepackage{algorithmic}

\newcommand{\algorithmicbreak}{\textbf{break}}
\newcommand{\Break}{\algorithmicbreak}

\usepackage{color}

\ifdefined\TeraoModifyRedColor
\newcommand{\tatsuya}[1]{\textcolor{red}{#1}}
\newcommand{\tatsuyadel}[1]{\tatsuya{\sout{{#1}}}}

\else
\newcommand{\tatsuya}[1]{#1}
\newcommand{\tatsuyadel}[1]{}

\fi

\title{Faster Matroid Partition Algorithms}
\author{Tatsuya Terao
\thanks{Kyoto University.
E-mail: ttatsuya@kurims.kyoto-u.ac.jp}
}

\date{}

\begin{document}






\maketitle

\begin{abstract}
In the matroid partitioning problem,
we are given $k$ matroids $\mathcal{M}_1 = (V, \mathcal{I}_1), \dots , \mathcal{M}_k = (V, \mathcal{I}_k)$ defined over a common ground set $V$ of $n$ elements,
and we need to find a partitionable set $S \subseteq V$ of largest possible cardinality, denoted by $p$.
Here, a set $S \subseteq V$ is called partitionable 
if there exists a partition $(S_1, \dots , S_k)$ of $S$ with $S_i \in \mathcal{I}_i$ for $i = 1, \ldots, k$.
In 1986, Cunningham~[SICOMP 1986] presented a matroid partition algorithm that uses $O(n p^{3/2} + k n)$ independence oracle queries,
which was the previously known best algorithm.
This query complexity is $O(n^{5/2})$ when $k \leq n$.


\tatsuya{Our main result is to present a matroid partition algorithm that uses $\tilde{O}(k'^{1/3} n p + k n)$ independence oracle queries, where $k' = \min\{k, p\}$.
This query complexity is $\tilde{O}(n^{7/3})$ when $k \leq n$, and this improves upon the one of previous Cunningham's algorithm.}
To obtain this, we present a new approach \emph{edge recycling augmentation}, which can be attained through new ideas:
an efficient utilization of the binary search technique by Nguy$\tilde{{\hat{\text{e}}}}$n~[2019] and Chakrabarty-Lee-Sidford-Singla-Wong~[FOCS 2019] and
a careful analysis of the independence oracle query complexity.
Our analysis differs significantly from the one for matroid intersection algorithms,
because of the parameter $k$.
We also present a matroid partition algorithm that uses $\tilde{O}((n + k) \sqrt{p})$ rank oracle queries.
\end{abstract}

\section{Introduction}

The \emph{matroid partitioning problem}\footnote{The \emph{matroid partitioning problem} is sometimes called simply \emph{matroid partition}. Matroid partition is also called \emph{matroid union} or \emph{matroid sum}.} is one of the most fundamental problems in combinatorial optimization.
The problem is sometimes introduced as an important matroid problem along with the \emph{matroid intersection problem}; see~\cite[Section 41--42]{schrijver2003combinatorial} and \cite[Section 13.5--6]{korte2011combinatorial}.
In the problem,
we are given $k$ matroids $\M_1 = (V, \mathcal{I}_1), \ldots, \M_k = (V, \mathcal{I}_k)$
defined over a common ground set $V$ of $n$ elements, 
and the objective is to find a partitionable set $S \subseteq V$ of largest possible cardinality, denoted by $p$.
Here, we call a set $S \subseteq V$ partitionable 
if there exists a partition $(S_1, \dots , S_k)$ of $S$ with $S_i \in \mathcal{I}_i$ for $i = 1, \ldots, k$.
This problem has a number of applications
such as
matroid base packing, 
packing and covering of trees and forests, 
and Shannon switching game.
There are much more applications; 
see~\cite[Section 42]{schrijver2003combinatorial}.


To design an algorithm for arbitrary matroids,
it is common to consider an oracle model:
an algorithm accesses a matroid through an oracle.
The most standard and well-studied oracle is an \emph{independence oracle},
which takes as input a set $S \subseteq V$ and outputs whether $S \in \mathcal{I}$ or not.
Some recent studies for fast matroid intersection algorithms also consider
a more powerful oracle called \emph{rank oracle},
which takes as input a set $S \subseteq V$ and outputs the size of the maximum cardinality independent subset of $S$.
In the design of efficient algorithms,
the goal is to minimize the number of such oracle accesses in a matroid partition algorithm.
We consider both independence oracle model and rank oracle model, and 
present the best query algorithms for both oracle models.


The matroid partitioning problem is closely related to the matroid intersection problem. 
Actually, the matroid partitioning problem and the matroid intersection problem are polynomially equivalent; 
see \cite{edmonds1970submodular,edmonds2003submodular}.

In the \emph{matroid intersection problem},
we are given two matroids $\M' = (V, \mathcal{I}')$ and $\M'' = (V, \mathcal{I}'')$
defined over a common ground set $V$ of $n$ elements, and
the objective is to find a common independent set $S \in \mathcal{I}' \cap \mathcal{I}''$ of largest possible cardinality,
denoted by $r$.

Starting the work of Edmonds \cite{edmonds1968matroid,edmonds1979matroid,edmonds2003submodular} in the 1960s,
algorithms with polynomial query complexity for the matroid intersection problem have been
studied \cite{aigner1971matching, lawler1975matroid, cunningham1986improved, lee2015faster, chakrabarty2019faster, blikstad2021breaking_STOC, blikstad2021breaking, blikstad2023fast}.
Nguy$\tilde{{\hat{\text{e}}}}$n~\cite{nguyen2019note} and Chakrabarty-Lee-Sidford-Singla-Wong~\cite{chakrabarty2019faster} independently presented a new excellent binary search technique 
that can find edges in the exchange graph
and presented a first combinatorial algorithm that uses $\tilde{O}(n r)$ independence oracle queries\footnote{The $\tO$ notation omits factors polynomial in $\log n$.}.
Chakrabarty et al.~\cite{chakrabarty2019faster} also presented a $(1 - \epsilon)$ approximation matroid intersection algorithm that uses $\tO(n^{1.5}/\epsilon^{1.5})$ independence oracle queries.
Blikstad-van den Brand-Mukhopadhyay-Nanongkai~\cite{blikstad2021breaking_STOC} developed a fast algorithm to solve a graph reachability problem, and
broke the $\tilde{O}(n^2)$-independence-oracle-query bound by combining this with previous exact and approximation algorithms.
Blikstad~\cite{blikstad2021breaking} improved the independence query complexity of the approximation matroid intersection algorithm.
This leads to 
a randomized matroid intersection algorithm that uses $\tilde{O}(n r^{3/4})$ independence oracle queries,
which is currently the best algorithm for the matroid intersection problem for the full range of $r$.
This also leads to 
a deterministic matroid intersection algorithm that uses $\tilde{O}(n r^{5/6})$ independence oracle queries,
which is currently the best deterministic algorithm for the matroid intersection problem for the full range of $r$.




We can solve the matroid partitioning problem by using a reduction to the matroid intersection problem.
A well-known reduction reduces the matroid partition to the matroid intersection
whose ground set size is $k n$.
Here, one of the input matroids of the obtained matroid intersection instance is the direct sum of $k$ matroids.
This leads to a matroid partition algorithm using too many independence oracle queries.
Even if we use the currently best algorithm for matroid intersection,
the naive reduction leads to a matroid partition algorithm that uses $\tO(k^2 n p^{3/4})$ independence oracle queries.
Since the matroid partition problem itself is an important problem with several applications,
it is meaningful to focus on the query-complexity of the matroid partitioning problem.

A direct algorithm for the matroid partitioning problem was first given by Edmonds in 1968 \cite{edmonds1968matroid}.
Algorithms with polynomial query complexity for the matroid partitioning problem have been
studied in the literature \cite{knuth1973matroid, greene1975some, cunningham1986improved, gabow1988forests, frank2012simple, blikstad2023fast, quanrud2023}.

Cunningham \cite{cunningham1986improved} designed a matroid partition algorithm that uses $O(n p^{3/2} + k n)$ independence oracle queries.
He uses a \emph{blocking flow} approach, 
which is similar to Hopcroft-Karp's bipartite matching algorithm or Dinic's maximum flow algorithm.
The independence query complexity of Cunningham's algorithm is $O(n^{5/2})$ when $k \leq n$.
Note that $p \leq n$ obviously holds.
This was the best algorithm for the matroid partitioning problem for nearly four decades.
We study faster matroid partition algorithms by using techniques that were recently developed for fast matroid intersection algorithms.

Our first result is the following theorem, which is obtained by combining Cunningham's technique
and the binary search technique by \cite{nguyen2019note, chakrabarty2019faster}.

\begin{theorem}[Details in Theorem \ref{thm:matroid_partition_ind_oracle}] \label{matroid_partition_ind_oracle}
There is an algorithm that uses $\tilde{O}(k n \sqrt{p})$ independence oracle queries and solves the matroid partitioning problem.
\end{theorem}

The independence query complexity of the algorithm given in Theorem~\ref{matroid_partition_ind_oracle} improves upon the one of Cunningham's algorithm~\cite{cunningham1986improved}
when $k$ is small.
However, when $k = \Theta(n)$, the independence query complexity of the algorithm given in Theorem~\ref{matroid_partition_ind_oracle} is $\tilde{O}(n^{5/2})$, which is not strictly less than the one in Cunningham's algorithm.



When $k \leq n$,
we sometimes bound the number of queries by a function on a single variable $n$, where we recall that $p \leq n$.
This makes it easy to compare the query complexity of different algorithms.

Our second result is to obtain an algorithm that 
is strictly better than Cunningham's algorithm even
when $k$ is large.

\begin{theorem}[Details in Theorem \ref{thm:faster_matroid_partition_ind_oracle_2}]  \label{faster_matroid_partition_ind_oracle}
\tatsuya{There is an algorithm that uses $\tilde{O}(k'^{1 / 3} n p + k n)$ independence oracle queries and solves the matroid partitioning problem, where $k' = min\{k, p\}$.}
\end{theorem}
This is the main contribution of this paper.
The independence query complexity of the algorithm given in Theorem~\ref{faster_matroid_partition_ind_oracle} improves the one of the algorithm given in Theorem \ref{matroid_partition_ind_oracle}
when $k = \omega(p^{3/4})$.
The independence query complexity of the algorithm given in Theorem~\ref{faster_matroid_partition_ind_oracle} is $\tilde{O}(n^{7/3})$ when $k \leq n$.
This improves the algorithm by Cunningham \cite{cunningham1986improved} and our algorithm given in Theorem \ref{matroid_partition_ind_oracle}.
It should be emphasized here that this is the first improvement since 1986.
\tatsuya{The previous and new bounds on the independence query complexity for all range of $k$ are summarized in Figure \ref{fig:complexity}.
As in this figure, we obtain a matroid partition algorithm faster than Cunningham's algorithm for all range of $k$ except for the tight value part.}
We note that this algorithm requires \tatsuya{$O(k'^{2/3} n p)$} time complexity other than independence oracle queries.


\begin{figure}[h]
\centering
\includegraphics[width=12cm]{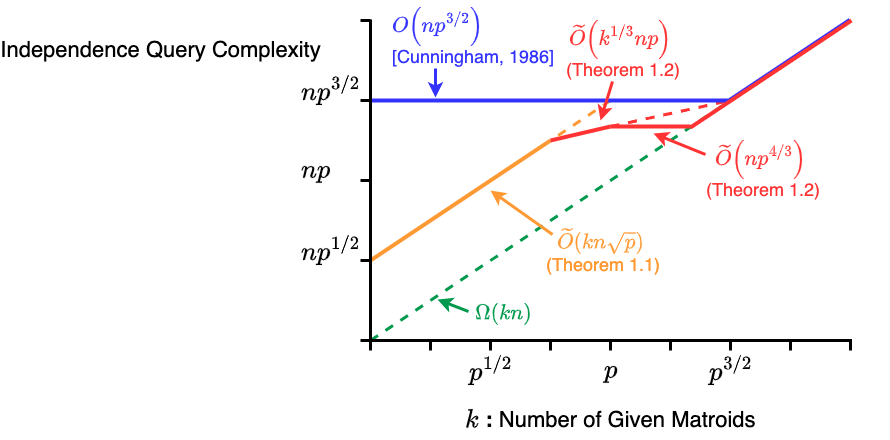}
\caption{State-of-the-art bounds on the independence query complexity for solving the matroid partitioning problem. 
The green dashed line represents an $\Omega(k n)$ lower bound;
see Section \ref{sec:conclusion} for details.}
\label{fig:complexity}
\end{figure}

We also consider the query complexity in the rank oracle model.
Note that the rank oracle is at least as powerful as the independence oracle.

\begin{theorem}[Details in Theorem \ref{thm:matroid_partition_rank_oracle}] \label{matroid_partition_rank_oracle}
There is an algorithm that uses $\tilde{O}((n + k) \sqrt{p})$ rank oracle queries and solves the matroid partitioning problem.
\end{theorem}
The rank query complexity of the algorithm given in Theorem~\ref{matroid_partition_rank_oracle} is $\tilde{O}(n^{3/2})$ when $k \leq n$.




\subsection{Technical Overview}

\paragraph*{Cunningham's matroid partition algorithm.}
The auxiliary graph called \emph{exchange graph} plays an important role in almost all combinatorial algorithms for matroid intersection.
In matroid intersection algorithms, we begin with an empty set and repeatedly increase the size of the independent set by augmenting along shortest paths in the exchange graph.
In the same way, 
Knuth~\cite{knuth1973matroid} and Greene-Magnanti~\cite{greene1975some} gave matroid partition algorithms by using the auxiliary graph with $O(n p)$ edges,
which we call \emph{compressed exchange graph}.

To improve the running time,
Cunningham~\cite{cunningham1986improved} developed \emph{blocking flow} approach for matroid partition and intersection,
which is akin to bipartite matching algorithm by Hopcroft-Karp~\cite{hopcroft1973n}.
The blocking flow approach is applied in each phase of the algorithm.
In Hopcroft-Karp's bipartite matching algorithm,
we find a maximal set of vertex-disjoint shortest paths and augment along these paths simultaneously.
In contrast to this, in a matroid partition algorithm, the augmentations can not be done in parallel,
since one augmentation can change the compressed exchange graph.
Cunningham showed that we can find multiple augmenting paths of the same length and run all the augmentations
in one phase.
In Cunningham's matroid partition algorithm, one phase uses only $O(n p)$ independence oracle queries 
(each edge is queried only once in one phase).

Cunningham showed that the number of different lengths of shortest augmenting paths during the algorithm is $O(\sqrt{p})$ and
then the number of phases is $O(\sqrt{p})$.
Therefore, Cunningham's matroid partition algorithm uses $O(n p^{3/2} + k n)$ independence oracle queries in total
(enumerating all edges entering sink vertices uses $O(k n)$ independence oracle queries).
We note that this query complexity is $O(n^{5/2})$ when $k \leq n$.

\paragraph*{Combining blocking flow approach and binary search subroutine.}
To develop the matroid partition algorithm, given in Theorem \ref{matroid_partition_ind_oracle}, that uses $\tO(k n \sqrt{p})$ independence oracle queries,
we combine the blocking flow approach proposed by Cunningham \cite{cunningham1986improved} and the binary search procedure proposed by \cite{nguyen2019note, chakrabarty2019faster}.
By using the binary search procedure,
we obtain an algorithm that uses $\tilde{O}(kn)$ independence oracle queries and 
performs a breadth first search in the compressed exchange graph.
We also obtain an algorithm that uses $\tilde{O}(kn)$ independence oracle queries and 
runs all the augmentations in a single phase.
Since Cunningham showed that the number of phases is $O(\sqrt{p})$, we can easily obtain a matroid partition algorithm that 
uses $\tO(k n \sqrt{p})$ independence oracle queries.
Our algorithm does not contain technical novelty
in a sense that this algorithm is obtained by simply
combining Cunningham's technique 
and the binary search technique by \cite{nguyen2019note, chakrabarty2019faster}.
Nevertheless,
this result is important in a sense that we improve the independence query complexity of a matroid partition algorithm.

\paragraph*{Edge recycling augmentation.}
\tatsuya{
In our breadth first search, we need to check, for all vertices $v \in V$ and all indices $i \in [k]$, whether there exists an edge from a vertex $v$ to some vertex in $S_i$ in the compressed exchange graph.
Note that the partition of the current partitionable set $S$ is $(S_1, \ldots, S_k)$.
Then, independence query complexity of a breadth first search of the compressed exchange graph
seems to be $\Omega(\min\{k n, n p \})$,
even if we use the binary search procedure; see Section \ref{sec:blockingflow} for details.
It is not clear whether we can develop a matroid partition algorithm that runs a breadth first search $o(\sqrt{p})$ times,
and so, algorithms by the blocking flow approach are now stuck at $\Omega(\min\{k, p \} n \sqrt{p})$ independence oracle queries.
}
In the setting where $k = \Theta(n)$ and $p = \Theta(n)$, algorithms by the blocking flow approach are stuck at $\Omega(n^{5/2})$ independence oracle queries
even if we use the excellent binary search procedure.

In order to break this $O(n^{5/2})$-independence-oracle-query bound,
we introduce a new approach \emph{edge recycling augmentation}
and develop a matroid partition algorithm whose independence query complexity is sublinear in $k$.
Then we present a matroid partition algorithm that uses $\tO(n^{7/3})$ independence oracle queries when $k \leq n$.

Our new approach edge recycling augmentation is 
applied in each phase of the algorithm
in the same way as the blocking flow approach.
In one phase of edge recycling augmentation, we first compute the edge set $E^*$ in the compressed exchange graph,
which uses $O(n p)$ independence oracle queries.
Then we simply repeat to run a breadth first search and find a shortest path in the compressed exchange graph.
This breadth first search is performed by using the information of $E^*$.
The precomputation of $E^*$ may seem too expensive
since we have the excellent binary search tool to find edges in the compressed exchange graph.
However, we can \emph{recycle} some edges in $E^*$ during the repetition of breadth first searches,
which plays an important role in an analysis of our new approach.
Note that, edge recycling augmentation runs a breadth first search before every augmentation,
while the blocking flow approach runs a breadth first search only once in the beginning of each phase.

Our crucial observation is that all edges entering a vertex in $S_i$ are not changed unless $S_i$ was updated by the augmentation.
Recall that the partition of the current partitionable set $S$ is $(S_1, \ldots, S_k)$.
Then, even after some augmentations, we can use $E^*$ to find edges entering a vertex $u \in S_i$ such that $S_i$ was not updated by the augmentation.
This observation is peculiar to the matroid partition.
In a breadth first search of the edge recycling augmentation approach, 
we use the binary search procedure only to find edges entering $u \in S_i$ such that $S_i$ was updated by the augmentation.
In one phase,
we repeat to run a breadth first search so that the total number of the binary search procedure calls is $O(n p)$.

We combine the blocking flow approach algorithm and the edge recycling augmentation approach algorithm.
\tatsuya{By a careful analysis of the independence query complexity, we obtain a matroid partition algorithm that uses $\tO(k'^{1/3} n p + k n)$ independence oracle queries, where $k' = \min \{k, p\}.$}

Note that this edge recycling augmentation approach differs significantly from existing fast matroid intersection algorithms.
The key technical contribution of this paper is to introduce this new approach.

\subsection{Related Work}
Blikstad-Mukhopadhyay-Nanongkai-Tu \cite{blikstad2023fast} introduced a new oracle model called dynamic oracle and 
developed a matroid partitioning algorithm that uses $\tO((n + r\sqrt{r}) \cdot \text{poly}(k))$ dynamic rank queries,
where $r = \max_i \max_{S_i \in \mathcal{I}_i} |S_i|$.
They also obtained an algorithm to solve the $k$-fold matroid union problem in $\tO(n \sqrt{r})$ time and dynamic rank queries,
which is the special case of the matroid partitioning problem where all matroids $\M_1, \dots, \M_k$ are identical.
Quanrud \cite{quanrud2023} developed an algorithm that solves the $k$-fold matroid union problem and uses $\tO(n^{3/2})$ independence oracle queries for the full range of $r$ and $k$.
Quanrud also considered the $k$-fold matroid union problem in the more general settings where the elements have integral and real-valued capacities.

For certain special matroids, faster matroid partition algorithms are known.
For linear matroids, 
Cunningham~\cite{cunningham1986improved} presented an $O(n^3 \log n)$-time algorithm that solves the matroid partitioning problem
on $O(n)$ matrices that have $n$ columns and at most $n$ rows.
For graphic matroids,
the $k$-forest problem is a special case of the matroid partitioning problem.
In the problem, we are given an undirected graph and a positive integer $k$, and
the objective is to find a maximum-size union of $k$ forests.
Gabow-Westermann~\cite{gabow1988forests} presented an $O(\min \{ k^{3/2}\sqrt{nm(m + n \log n)}, k^{1/2}m\sqrt{m + n \log n},$ $kn^2 \log k, \frac{m^2}{k}\log k \} )$-time algorithm to solve the $k$-forest problem, 
where $n$ and $m$ denote the number of vertices and edges, respectively.
Blikstad et al.~\cite{blikstad2023fast} and Quanrud \cite{quanrud2023} independently obtained an $\tO(m + (k n)^{3/2})$ time algorithm to solve the $k$-forest problem.


Kawase-Kimura-Makino-Sumita~\cite{kawase2021optimal} studied matroid partitioning problems for various objective functions.

For the weighted matroid intersection,
Huang-Kakimura-Kamiyama~\cite{huang2016exact} developed a technique that transforms any unweighted matroid intersection algorithm into an algorithm that solves the weighted case with an $O(W)$ factor.
Huang et al. also presented a $(1 - \epsilon)$ approximation weighted matroid intersection algorithm that uses $\tO(nr^{3/2}/\epsilon)$ independence oracle queries.
Chekuri-Quanrud~\cite{chekuri2016fast} improved the independence query complexity and presented a $(1 - \epsilon)$ approximation weighted matroid intersection algorithm that uses $O(n r/\epsilon^2)$ independence oracle queries,
which can be improved by applying a recent faster approximation unweighted matroid intersection algorithm by Blikstad~\cite{blikstad2021breaking}.
Tu~\cite{u2022subquadratic} gave a weighted matroid intersection algorithm that uses $\tO(n r^{3/4} \log W)$ rank oracle queries,
which also uses the binary search procedure proposed by \cite{nguyen2019note, chakrabarty2019faster}.

For matroids of rank $n / 2$, Harvey~\cite{harvey2008matroid} showed a lower bound of $(\log_2 3)n - o(n)$ independence oracle queries for matroid intersection.
Blikstad et al.~\cite{blikstad2023fast} showed super-linear $\Omega(n \log n)$ query lower bounds for matroid intersection and partitioning problem in their dynamic-rank-oracle and the independence oracle models.

\subsection{Paper Organization}
In Section \ref{sec:preliminaries}, we introduce the notation and the known results for matroid partition and intersection.
Next, in Section \ref{sec:blockingflow}, we present our matroid partition algorithm using the blocking flow approach.
Then, in Section \ref{sec:faster_algorithm_for_large_k}, we present our new approach edge recycling augmentation and our faster matroid partition algorithm for large $k$.

\section{Preliminaries}\label{sec:preliminaries}
\subsection{Matroids}

\paragraph*{Notation.}
For a positive integer $a$, we denote $[a] = \{1, \dots , a\}$.
For a finite set $X$, let $\# X$ and $|X|$ denote the cardinality of $X$, which is also called the {\em size} of $X$. 
We will often write $A + x := A \cup \{ x \}$ and $A - x := A \setminus \{ x \}$.
We will also write $A + B := A \cup B$ and $A - B := A \setminus B$,
when no confusion can arise.



\paragraph*{Matroid.}
A pair $\M = (V, \mathcal{I})$ for a finite set $V$ and non-empty $\mathcal{I} \subseteq 2^{V}$ is called a {\em matroid} if the following property is satisfied.

(Downward closure) if $S \in \mathcal{I}$ and $S' \subseteq S$, then $S' \in \mathcal{I}$.

(Augmentation property) if $S, S' \in \mathcal{I}$ and $|S'| < |S|$, then there exists $x \in S \setminus S'$ such that $S' + x \in \mathcal{I}$.

A set $S \subseteq V$ is called {\em independent} if $S \in \mathcal{I}$ and {\em dependent} otherwise.

\paragraph*{Rank.}
For a matroid $\M = (V, \mathcal{I})$, we define the {\em rank} of $\M$ as ${\rm rank}(\M) = \max \{ |S| \mid S \in \mathcal{I} \}$.
In addition, for any $S \subseteq V$, we define the {\em rank} of $S$ as ${\rm rank}_{\M}(S) = \max \{ |T| \mid T \subseteq S, T \in \mathcal{I} \}$.



\paragraph*{Matroid Intersection.}
Given two matroids $\mathcal{M}' = (V, \mathcal{I}')$, $\mathcal{M}'' = (V, \mathcal{I}'')$, 
we define their \emph{intersection} by $(V, \mathcal{I}' \cap \mathcal{I}'' )$.
The \emph{matroid intersection problem} asks us to find the largest common independent set, whose cardinality we denote by $r$.
Note that the intersection of matroids is not a matroid in general and 
the problem to find a maximum common independent set of more than two matroids is {\rm NP}-hard.



\paragraph*{Matroid Partition (Matroid Union).}
Given $k$ matroids $\mathcal{M}_1 = (V, \mathcal{I}_1), \dots, \mathcal{M}_k = (V, \mathcal{I}_k)$, $S \subseteq V$ is called \emph{partitionable} 
if there exists a partition $(S_1, \dots , S_k)$ of $S$ such that $S_i \in \mathcal{I}_i$ for $i \in [k]$.
The \emph{matroid partitioning problem} asks us to find the largest partitionable set, whose cardinality we denote by $p$.
Let $\tilde{\mathcal{I}}$ be the family of partitionable subset of $V$.
Then, $(V, \tilde{\mathcal{I}})$ is called the {\em union} or {\em sum} of $k$ matroids $\mathcal{M}_1 \dots, \mathcal{M}_k$.
Note that Nash-Williams Theorem~\cite{nash1966application} states that the union $(V, \tilde{\mathcal{I}})$ of the $k$ matroids is also a matroid.

\paragraph*{Oracles.}
Throughout this paper, we assume that we can access a matroid $\M = (V, \I)$ only through an oracle.
Given a subset $S \subseteq V$, 
an \emph{independence} oracle outputs whether $S \in \I$ or not.
Given a subset $S \subseteq V$, 
a \emph{rank} oracle outputs ${\rm rank}_{\M}(S)$.
Since one query of the rank oracle can determine whether a given subset is independent,
the rank oracle is more powerful than the independence oracle.

\paragraph*{Binary Search Technique.}
Chakrabarty-Lee-Sidford-Singla-Wong~\cite{chakrabarty2019faster} showed that the following procedure can be implemented efficiently by using binary search in the independence oracle model.
(This was developed independently by Nguy$\tilde{{\hat{\text{e}}}}$n~\cite{nguyen2019note}.)
Given a matroid $\M = (V, \I)$, an independent set $S \in \I$, an element $v \in V \setminus S$, and $B \subseteq S$, the objective is to find
an element $u \in S$ that is exchangeable with $v$ (that is, $S + v - u \in \I$) or conclude there is no such an element.
We skip the proof in this paper; see \cite[Section 3]{chakrabarty2019faster} for a proof.



\begin{lemma}[Edge Search via Binary search, \cite{nguyen2019note,chakrabarty2019faster}]\label{lem:binary_search_matroid_partition_ind}
There exists an algorithm \textup{\texttt{FindOutEdge}} which, 
given a matroid $\M = (V, \mathcal{I})$, an independent set $S \in \mathcal{I}$, an element $v \in V \setminus S$, and $B \subseteq S$, finds an element $u \in B$ such that $S + v - u \in \mathcal{I}$ or otherwise determine that no such element exists, and uses $O(\log |B|)$ independence queries.
\tatsuya{Furthermore, if there exists no such an element, then the procedure \textup{\texttt{FindOutEdge}} uses only one independence query.\footnote{Since there exists such an element if and only if $S + v - B \in \mathcal{I}$.}}
\end{lemma}


In the same way as Lemma~\ref{lem:binary_search_matroid_partition_ind}, Chakrabarty et al.~\cite{chakrabarty2019faster} also provided a technique which efficiently finds
an exchangeable element or conclude there is no such an element
in the rank oracle model; see \cite[Section 3]{chakrabarty2019faster} for a proof.
 
\begin{lemma}[Edge Search via Binary search, \cite{chakrabarty2019faster}]\label{lem:binary_search_matroid_partition_rank}
There exists an algorithm \textup{\texttt{FindInEdge}} which,
given a matroid $\M = (V, \mathcal{I})$, an independent set $S \in \mathcal{I}$, an element $u \in S$, and a subset $A \subseteq V \setminus S$, finds an element $v \in A$ such that $S + v - u \in \mathcal{I}$ or otherwise determine that no such an element exists, and uses $O(\log |A|)$ rank queries.
\end{lemma}


\subsection{Techniques for Matroid Intersection}\label{sec:ape:matroid_intersection}

Here we provide known results about the matroid intersection.

\begin{definition}[Exchange Graph]
Consider a common independent set $S \in \mathcal{I}' \cap \mathcal{I}''$.
The {\em exchange graph} is defined as a directed graph $G(S) = (V \cup \{s, t \}, E)$, with $s, t \notin V$ and $E = E' \cup E'' \cup E_s \cup E_t$, where
\begin{align*}
E' = & \{(u, v) \mid u \in S, v \in V \setminus S, S - u + v \in \mathcal{I}' \} ,\\
E'' = & \{(v, u) \mid u \in S, v \in V \setminus S, S - u + v \in \mathcal{I}'' \} ,\\
E_s = & \{(s, v) \mid v \in V \setminus S, S + v \in \mathcal{I}' \} , \text{and} \\
E_t = & \{(v, t) \mid v \in V \setminus S, S + v \in \mathcal{I}'' \} .
\end{align*}
\end{definition}

\begin{lemma}[Shortest Augmenting Path {\cite[Theorem 41.2]{schrijver2003combinatorial}}]
Let $s, v_1, v_2, \dots , v_{l - 1}, t$ be a shortest $(s, t)$-path in the exchange graph $G(S)$ relative to a common independent set $S \in \I' \cap \I''$.
Then $S' = S + v_1 - v_2 + \cdots - v_{l - 2} + v_{l - 1} \in \I' \cap \I''$.
\end{lemma}

In a matroid intersection algorithm, 
we begin with an empty set $S$.
Then we repeat to find an augmenting path in the exchange graph $G(S)$ and to update the current set $S$.
If there is no $(s, t)$-path in the exchange graph $G(S)$, then $S$ is a common independent set of maximum size.
If there is an $(s, t)$-path in the exchange graph $G(S)$, then we pick a shortest path and obtain a common independent set $S' \in \I' \cap \I''$ of $|S| + 1$ elements.

Cunningham's matroid intersection algorithm \cite{cunningham1986improved} and recent faster matroid intersection algorithms \cite{chakrabarty2019faster,nguyen2019note,blikstad2021breaking_STOC,blikstad2021breaking}
rely on the following two lemmas.

\begin{lemma}[\cite{cunningham1986improved}]\label{lem:total_bound_matroid_intersection}
For any two matroids $\M' = (V, \mathcal{I}')$ and $\M'' = (V, \mathcal{I}'')$, 
if the length of the shortest augmenting path in exchange graph $G(S)$ relative to a common independent set $S \in \mathcal{I}' \cap \mathcal{I}''$ is at least $d$, then $|S| \geq (1 - \frac{O(1)}{d}) \cdot r$, where $r$ is the size of a largest common independent set.
\end{lemma}



\begin{lemma}[Monotonicity Lemma~\cite{cunningham1986improved, haselmayr2008, price2015,chakrabarty2019faster}]\label{lem:monotonicity_matroid_intersection}
For any two matroids $\M' = (V, \mathcal{I}')$ and $\M'' = (V, \mathcal{I}'')$, 
suppose we obtain a common independent set $S' \in \mathcal{I}' \cap \mathcal{I}''$ by augmenting $S \in \mathcal{I}' \cap \mathcal{I}''$ along a shortest augmenting path in $G(S)$.
Note that $|S'| > |S|$.
Let $d$ denote the distance in $G(S)$ and $d'$ denote the distance in $G(S')$.
Then for all $v \in V$,
\begin{itemize}
 \item[(i)] If $d(s, v) < d(s, t)$, then $d'(s, v) \geq d(s, v)$. If $d(v, t) < d(s, t)$, then $d'(v, t) \geq d(v, t)$.
 \item[(ii)] If $d(s, v) \geq d(s, t)$, then $d'(s, v) \geq d(s, t)$. If $d(v, t) \geq d(s, t)$, then $d'(v, t) \geq d(s, t)$.
\end{itemize}
\end{lemma}

 

\subsection{Compressed Exchange Graph for Matroid Partition}\label{subsec:compressed_exchange_graph}
The matroid partitioning problem can be solved by a matroid intersection algorithm.
Let $\hat{V} = V \times [k]$, and define

\begin{align*}
\hat{\mathcal{I}}' = & \{ \hat{I} \subseteq \hat{V} \mid \forall v \in V, \#\{ i \in [k] \mid (v, i) \in \hat{I} \} \leq 1 \}, \\
\hat{\mathcal{I}}'' = & \{ \hat{I} \subseteq \hat{V} \mid \forall i \in [k], \{ v \in V \mid (v, i) \in \hat{I} \} \in \mathcal{I}_i \} .
\end{align*}

Then, $\hat{\M}' = (\hat{V}, \hat{\mathcal{I}}')$ is a partition matroid.
Since $\hat{\M}'' = (\hat{V}, \hat{\mathcal{I}}'')$ is the direct sum of matroids $(V, \mathcal{I}_i)$ for all $i \in [k]$, it is also a matroid.
Then, the family of partitionable subsets of $V$ can be represented as 

\begin{equation*}
\{ S \subseteq V \mid \exists \pi \colon S \to [k], \{(v, \pi(v)) \mid v \in S \} \in \hat{\mathcal{I}}' \cap \hat{\mathcal{I}}'' \} .
\end{equation*}

Therefore, we can solve the matroid partitioning problem by computing a common independent set of maximum size in $\mathcal{I}'$ and $\mathcal{I}''$.
However, we might use too many independence oracle queries when solving the matroid partitioning problem by using this reduction to the matroid intersection problem.
This is due to the following reasons.
When solving the matroid intersection problem that was reduced by the matroid partitioning problem,
the size of the ground set of that matroid intersection problem is $O(k n)$, and then
the number of edges in the exchange graph is $O(k n p)$, 
which depends heavily on $k$.
Furthermore, since we consider the total query complexity of the independence oracle of each matroid $\M_i = (V, \mathcal{I}_i)$ for all $i \in [k]$,
the query complexity of the independence query of the matroid $\hat{\M}'' = (\hat{V}, \hat{\mathcal{I}}'')$ also depends heavily on $k$.

Then, to improve the running time, Knuth~\cite{knuth1973matroid} and Greene-Magnanti~\cite{greene1975some} gave a matroid partition algorithm that uses the following auxiliary graph with $O(n p)$ edges,
which we call {\em compressed exchange graph}.

\begin{definition}[Compressed Exchange Graph~\cite{knuth1973matroid,greene1975some, cunningham1986improved}]\label{def:compressed_exchange_graph}
Consider a partition $(S_1, \dots , S_k)$ of $S \subseteq V$ such that $S_i \in \mathcal{I}_i$ for all $i \in [k]$.
The \emph{compressed exchange graph} is defined as a directed graph $G(S_1, \dots , S_k) = (V \cup \{s, t_1, \dots, t_k\}, E)$, with $s, t_1, \dots, t_k \notin V$ and $E = E' \cup E_s \cup E_t$, where
\begin{align*}
E' = & \{(v, u) \mid \exists i \in [k], u \in S_i, S_i + v \notin \mathcal{I}_i , S_i + v - u \in \mathcal{I}_i \} ,\\
E_s = & \{(s, v) \mid v \in V \setminus S \} , \text{and} \\
E_t = & \bigcup_{i = 1}^{k} \{(v, t_i) \mid v \in V \setminus S_i , S_i + v \in \mathcal{I}_i \} .
\end{align*}
We denote $T = \{t_1, \dots, t_k\}$.
\end{definition}

In the matroid partition algorithm, 
we begin with an empty set $S$ and initialize $S_i = \emptyset$ for all $i \in [k]$.
If there is no vertex in $T$ which is reachable from $s$ in the compressed exchange graph $G(S_1, \dots, S_k)$, then $S$ is a partitionable set of maximum size.
If there is a path from $s$ to $T$ in the compressed exchange graph, then we pick a shortest path $s, v_1, \dots, v_{l - 1}, t_j$. 
Then we can obtain a partitionable set $S' = S + v_1$ and a partition $(S'_1, \dots , S'_k)$ of $S'$ such that $S'_i \in \mathcal{I}_i$ for all $i \in [k]$.
The validity of the algorithm follows from the following two lemmas, which we use throughout this paper.
Cunningham~\cite{cunningham1986improved} showed these lemmas by using the equivalence of the compressed exchange graph for the matroid partition and the exchange graph for the reduced matroid intersection;
see~\cite[Theorem 42.4]{schrijver2003combinatorial} for a direct proof that does not use the reduction to the matroid intersection.

\begin{lemma}\label{lem:augmenting_path_matroid_partition_2}
Given a partition $(S_1, \dots, S_k)$ of $S$ such that $S_i \in \mathcal{I}_i$ for all $i \in [k]$, 
there exists a partitionable set $S'$ whose size is at least $|S| + 1$
if and only if there is a vertex $t_j \in T$ that is reachable from $s$ in the compressed exchange graph $G(S_1, \dots, S_k)$.
\end{lemma}

\begin{lemma}[Shortest Augmenting Path]\label{lem:augmenting_path_matroid_partition}
Let $s, v_1, v_2, \dots, v_{l - 1}, t_j$ be a shortest $(s, T)$-path in the compressed exchange graph $G(S_1, \dots , S_k)$. 
Then $S' = S + v_1$ is a partitionable set.
\end{lemma}

We can construct a partition $(S'_1, \dots , S'_k)$ of $S'$ from a partition $(S_1, \dots , S_k)$ of $S$ and an augmenting path in the compressed exchange graph 
by using the following procedure \texttt{Update} (Algorithm \ref{alg:Update}).

\begin{algorithm}[H]
    \KwInput{a partition $(S_1, \dots , S_k)$ of $S$ $(\subseteq V)$ such that $S_i \in \mathcal{I}_i$ for all $i \in [k]$, and
    an augmenting path $s, v_1, \dots, v_{l - 1}, t_j$ }
    \KwOutput{a partition $(S'_1, \dots , S'_k)$ of $S'$ $(\subseteq V)$ such that $S'_i \in \mathcal{I}_i$ for all $i \in [k]$ and $S' = S + v_1$}
    For all $i \in [k]$, set $S'_i \gets S_i$ \\
    For all $v \in S$, denote by $\pi(v)$ the index such that $v \in S_{\pi(v)}$ \\
    \For{$i \in [l - 2]$} {
        $S'_{\pi(v_{i + 1})} \gets S'_{\pi(v_{i + 1})} + v_i - v_{i + 1}$ \\
    }
    $S'_{j} \gets S'_j + v_{l - 1}$ \\
    \Return a partition $(S'_1, \dots , S'_k)$ of $S'$ 
    \caption{\texttt{Update}}\label{alg:Update}
\end{algorithm}


Cunningham~\cite{cunningham1986improved} observes that the equivalence between the exchange graph for the matroid intersection of two matroids $\hat{\M}' = (\hat{V}, \hat{\mathcal{I}}')$ and $\hat{\M}'' = (\hat{V}, \hat{\mathcal{I}}'')$ and
the compressed exchange graph for the matroid partition of $k$ matroids $(V, \mathcal{I}_1), \dots, (V, \mathcal{I}_k)$
to prove Lemmas \ref{lem:total_bound_matroid_partition} and \ref{lem:monotonicity_matroid_partition} and
to develop an efficient algorithm for matroid partition that employs the blocking flow approach.
For a fixed partition $(S_1, \dots, S_k)$ of $S$ and an element $v \in S$, let $\pi(v)$ be the index such that $v \in S_{\pi(v)}$.
We also denote by $\hat{S}$ the set $\{(v, \pi(v)) \in \hat{V} \mid v \in S \}$.
A path $s, v_1, v_2, \dots, v_{l - 1}, t_j$ in the compressed exchange graph for the matroid partition corresponds to a path $s, (v_1, \pi(v_2)), (v_2, \pi(v_2)), (v_2, \pi(v_3)), \dots,$ $ (v_{l - 1}, \pi(v_{l - 1})), (v_{l - 1}, j), t$ in the exchange graph for the matroid intersection.
Then, for all elements $v \in S$, we have $d_{G(\Ss)}(s, v) = 1 + \frac{1}{2} d_{G(\hat{S})}(s, (v, \pi(v)))$ and
$d_{G(\Ss)}(v, T) = \frac{1}{2} d_{G(\hat{S})}((v, \pi(v)), t)$.
We also have $d_{G(\Ss)}(s, T) = 1 + \frac{1}{2} d_{G(\hat{S})}(s, t)$.



Cunningham~\cite{cunningham1986improved} uses the following two lemmas to develop an efficient matroid partition algorithm by using blocking flow approach.
These lemmas can be shown from the correspondence between the exchange graph and the compressed exchange graph.
Note that Lemmas \ref{lem:total_bound_matroid_partition} and \ref{lem:monotonicity_matroid_partition} correspond to Lemmas \ref{lem:total_bound_matroid_intersection} and \ref{lem:monotonicity_matroid_intersection}, respectively.
We also use these two lemmas in our fast matroid partition algorithms.

\begin{lemma}[\cite{cunningham1986improved}]\label{lem:total_bound_matroid_partition}
Given a partition $(S_1, \dots , S_k)$ of $S$ such that $S_i \in \mathcal{I}_i$ for all $i \in [k]$.
If the length of a shortest augmenting path in the compressed exchange graph $G(S_1, \dots , S_k)$ is at least $d$, then $|S| \geq (1 - \frac{O(1)}{d}) \cdot p$, where $p$ is the size of a largest partitionable set.
\end{lemma}

\begin{lemma}[Monotonicity Lemma \cite{cunningham1986improved, haselmayr2008, price2015, chakrabarty2019faster}]\label{lem:monotonicity_matroid_partition}
Suppose we obtain a partition $(S'_1, \dots, S'_k)$ of $S'$ by augmenting a partition $(\Ss)$ of $S$ along a shortest augmenting path in $G(\Ss)$.
Note that $|S'| > |S|$.
Let $d$ denote the distance in $G(\Ss)$ and $d'$ denote the distance in $G(S'_1, \dots, S'_k)$.
Then for all $v \in V$,
\begin{itemize}
 \item[(i)] If $d(s, v) < d(s, T)$, then $d'(s, v) \geq d(s, v)$. If $d(v, T) < d(s, T)$, then $d'(v, T) \geq d(v, T)$.
 \item[(ii)] If $d(s, v) \geq d(s, T)$, then $d'(s, v) \geq d(s, T)$. If $d(v, T) \geq d(s, T)$, then $d'(v, T) \geq d(s, T)$.
\end{itemize}
\end{lemma}

As we will see later, we use the binary search technique given in Lemma \ref{lem:binary_search_matroid_partition_ind}
to find edges in the compressed exchange graph under the independence oracle model.
Note that the procedure \texttt{FindOutEdge}$(\M_i, S_i, v, B)$ gives us an efficient way to find edges from the vertex $v$ to a vertex $u \in B (\subseteq S_i)$ in the compressed exchange graph.

We also use the binary search technique given in Lemma \ref{lem:binary_search_matroid_partition_rank} 
to find edges in the compressed exchange graph under the rank oracle model.
Note that the procedure \texttt{FindInEdge}$(\M_i, S_i, u, A)$ gives us an efficient way to find edges from a vertex $v \in A$ to the vertex $u \in S_i$ in the compressed exchange graph.

In the same way as Lemma \ref{lem:binary_search_matroid_partition_rank},
given a matroid $\M = (V, \mathcal{I})$, an independent set $S \in \mathcal{I}$ and a subset $A \subseteq V \setminus S$, we can find an element $v \in A$ such that $S + v \in \mathcal{I}$ or otherwise determine that no such element exists by using $O(\log |A|)$ rank queries.
Note that the procedure \texttt{FindInEdge}$(\M_j, S_j, \emptyset, A)$ gives us an efficient way to find edges from $v \in A$ to the vertex $t_j \in T$ in the compressed exchange graph.

\section{Blocking Flow Algorithm} \label{sec:blockingflow}

In this section, we provide our matroid partition algorithms in both independence oracle  and rank oracle models, 
which is obtained by simply combining the \emph{blocking flow} approach proposed by Cunningham \cite{cunningham1986improved} and the binary search search procedure proposed by \cite{nguyen2019note, chakrabarty2019faster}.
In Sections \ref{subsec:blockingflow_ind} and \ref{subsec:blockingflow_rank},
we provide our matroid partition algorithms in independence oracle and rank oracle models, respectively.



\subsection{Blocking Flow Algorithm using Independence Oracle} \label{subsec:blockingflow_ind}

In this subsection we present our matroid partition algorithm using the blocking flow approach in the independence oracle model.
We show the following theorem, which implies Theorem \ref{matroid_partition_ind_oracle}.

\begin{theorem}\label{thm:matroid_partition_ind_oracle}
\tatsuya{There is an algorithm that uses $O\left((k n + p \log p) \sqrt{p}\right)$ independence oracle queries and solves the matroid partitioning problem.}
\end{theorem}
This result improves upon the previously known matroid partition algorithm by Cunningham~\cite{cunningham1986improved}
when $k = o(p)$.

For the proof, 
we first provide the procedure \texttt{GetDistanceIndependence} (Algorithm \ref{alg:getdistanceind}) that efficiently finds distances from $s$ to every vertex in the compressed exchange graph.
This algorithm simply runs a breadth first search by using the procedure \texttt{FindOutEdge}.

\begin{lemma}[Breadth First Search using Independence Oracle]\label{lem:getdistanceind}
Given a partition $(S_1, \dots , S_k)$ of $S$ $(\subseteq V)$ such that $S_i \in \mathcal{I}_i$ for all $i\in [k]$, 
the procedure $\textup{\texttt{GetDistanceIndependence}}$ (Algorithm\ref{alg:getdistanceind}) outputs $d \in \mathbb{R}^{V \cup \{s\} \cup T}$ such that, for $v \in V \cup \{s\} \cup T$, 
$d_v$ is the distance from $s$ to $v$ in the compressed exchange graph $G(S_1, \dots , S_k)$.
\tatsuya{The procedure $\textup{\texttt{GetDistanceIndependence}}$ uses $O\left(k n + p \log p\right)$ independence oracle queries.}
\end{lemma}
\begin{proof}
The procedure \texttt{GetDistanceIndependence} simply performs a breadth first search in the compressed exchange graph $G(S_1, \dots , S_k)$.
Thus, the procedure \texttt{GetDistanceIndependence} correctly computes distances from $s$ in $G(S_1, \dots , S_k)$.
Note that each vertex $v \in V$ is added to $Q$ at most once and each vertex $v \in S$ is removed from $B_{\pi(v)}$ at most once.
Thus, the number of independence oracle queries used in Line \ref{code:BFS_ind_1} is $O(k n)$.
The number of \texttt{FindOutEdge} calls that do not output $\emptyset$ is $O(p)$, and the number of \texttt{FindOutEdge} calls that output $\emptyset$ is $O(k n)$.
\tatsuya{Hence, by Lemma~\ref{lem:binary_search_matroid_partition_ind}, the number of independence oracle queries used in Line \ref{code:BFS_ind_2} is $O\left(k n + p \log p\right)$, which completes the proof.}
\end{proof}

\begin{algorithm}[H]
    \KwInput{a partition $(S_1, \dots , S_k)$ of $S$ $(\subseteq V)$ such that $S_i \in \mathcal{I}_i$ for all $i \in [k]$}
    \KwOutput{$d \in \mathbb{R}^{V \cup \{s\} \cup T}$ such that for $v \in V \cup \{s\} \cup T$, $d_v$ is the distance from $s$ to $v$ in $G(S_1, \dots , S_k)$}
    $d_s \gets 0$ \\
    For all $v \in V \setminus S$ let $d_v \gets 1$ \\
    For all $v \in S$ let $d_v \gets \infty$ \\
    For all $i \in [k]$ let $d_{t_i} \gets \infty$ \\
    $Q \gets \{ v \in V \setminus S \}$ \tcp{$Q :$ queue} 
    For all $i \in [k]$ let $B_i \gets S_i$ \\
    \While{$Q \neq \emptyset$}{
        Let $v$ be the element added to $Q$ earliest \\
        $Q \gets Q - v$ \\
        \For{$i \in [k]$ with $d_{t_i} = \infty$} {
            \If{$v \notin S_i$ and $S_i + v \in \mathcal{I}_i$}{ \label{code:BFS_ind_1}
                $d_{t_i} \gets d_v + 1$ \\
            }
        }
        \For{$i \in [k]$ with $v \notin S_i$}{
            \While{$u$ = \textup{\texttt{FindOutEdge}}($\M_i, S_i, v, B_i$) satisfies $u \neq \emptyset$} { \label{code:BFS_ind_2}
                $Q \gets Q + u$ \\
                $d_u \gets d_v + 1$ \\
                $B_i \gets B_i - u$
            }
        }
    }
    \Return {$d$}
    \caption{\texttt{GetDistanceIndependence}}\label{alg:getdistanceind}
\end{algorithm}

Next we provide our augmentation subroutine for our faster matroid partition algorithm.
We implement Cunningham's~\cite{cunningham1986improved} blocking flow approach for matroid partition 
by using the binary search procedure proposed by \cite{nguyen2019note, chakrabarty2019faster}.
This algorithm is similar to the matroid intersection algorithm of \cite{chakrabarty2019faster} in the rank oracle model.
The implementation is described as \texttt{BlockFlowIndependence} (Algorithm \ref{alg:blockflowind}).

\begin{algorithm}[H]
    \KwInput{a partition $(S_1, \dots , S_k)$ of $S$ $(\subseteq V)$ such that $S_i \in \mathcal{I}_i$ for all $i \in [k]$}
    \KwOutput{a partition $(S'_1, \dots , S'_k)$ of $S'$ $(\subseteq V)$ such that $S'_i \in \mathcal{I}_i$ for all $i \in [k]$, $|S'| > |S|$, and $d_{G(S_1', \dots, S_k')}(s, T) \geq d_{G(S_1, \dots, S_k)}(s, T) + 1$, or a partition $(S_1, \dots , S_k)$ of $S$ if no such $S'$ exists.}
    $d \gets \text{\texttt{GetDistanceIndependence}}(S_1, \dots , S_k)$ \label{line:blockflowind1} \\
    \lIf{$d_T = \infty$} {
        \Return $(S_1, \dots , S_k)$
    }
    For all $v \in V$ let $I_v \gets [k]$ \\
    For all $i \in [d_T - 1]$ let $L_i \gets \{ v \in V \mid d_v = i \}$ \\
    \While{$L_i \neq \emptyset$ for all $i \in [d_T - 1]$} {
        $l \gets 0$ , $a_{d_T} \gets \emptyset$ \\
        \While{$l < d_{T}$} {
            $a_{l + 1} \gets \emptyset$ \\
            \If{$l < d_T - 1$} {
                \If{$l = 0$} {
                    Pick arbitrary $v \in L_1$ \\
                    $a_1 \gets v$ \\
                } \Else {
                    \While{$I_{a_{l}} \neq \emptyset$}{
                        Pick arbitrary $i \in I_{a_{l}}$ \\
                        $a_{l + 1} \gets $ \texttt{FindOutEdge}($\M_i, S_i, a_l, L_{l + 1} \cap S_i$) \label{line:blockflowind3} \\
                        \lIf{$a_{l + 1} = \emptyset$} {
                            $I_{a_l} \gets I_{a_l} - i$
                        }
                        \lElse {
                            \Break 
                        }
                    }
                }
            } \Else {
                \For{$i \in [k]$} {
                    \If{$a_l \notin S_i$ and $S_i + a_l \in \mathcal{I}_i$} { \label{line:blockflowind2}
                        $a_{l + 1} \gets t_i$ \\
                        \Break \\
                    }
                }
            }

            \If{$a_{l + 1} = \emptyset$} {
                $L_l \gets L_l - a_l$ \\
                \lIf{$L_l = \emptyset$} {
                    \Break
                }
                $l \gets l - 1$ \\
            } \Else {
                $l \gets l + 1$ \\
            }
        }

        \If{$a_{d_T} \neq \emptyset$} {
            Denote by $P$ the path $s, a_1, \ldots , a_{d_T - 1}, a_{d_T}$ \\
            $(S_1, \dots , S_k) \gets \text{\texttt{Update}}((S_1, \dots , S_k), P)$ \label{line:blockflowind4} \\
            For all $i \in [d_T - 1]$ let $L_i \gets L_i - a_i$ \\
        }
    }
    \Return $(S_1, \dots , S_k)$
    \caption{\texttt{BlockFlowIndependence}}\label{alg:blockflowind}
\end{algorithm}

In the procedure \texttt{BlockFlowIndependence}, given a partition $(\Ss)$ of $S$, 
we first compute the distances from $s$ to every vertex in the compressed exchange graph $G(\Ss)$ using \texttt{GetDistanceIndependence} (Algorithm \ref{alg:getdistanceind}).
By using these distances, we divide $V$ into sets $L_1, L_2, \dots$, where each $L_i$ has all vertices $v$ such that the distance from $s$ to $v$ is $i$ in the compressed exchange graph $G(\Ss)$.
Then we search a path $s, a_1, a_2, \dots , a_{d_T - 1}, a_{d_T}$ in the compressed exchange graph $G(\Ss)$, where $a_i \in L_i$ for all $i \in [d_T - 1]$.
Note that we write $d_T = \min (d_{t_1}, \dots , d_{t_k})$.
If we find such a path, we augment a partition $(\Ss)$ of $S$ and remove $a_i$ from $L_i$ for all $i \in [d_T - 1]$.
Then we search a new path again
until no $(s, T)$-path of length $d_T$ can be found.
During the search for such a path, if the procedure concludes that some vertex in $L_i$ is not on such a path, then it removes the vertex from $L_i$.

\begin{lemma}[Blocking Flow using Independence Oracle] \label{lem:blockflowind}
Given a partition $(S_1, \dots , S_k)$ of $S$ $(\subseteq V)$ such that $S_i \in \mathcal{I}_i$ for all $i\in [k]$, 
the procedure \textup{\texttt{BlockFlowIndependence}} (Algorithm \ref{alg:blockflowind}) outputs 
a partition $(S'_1, \dots , S'_k)$ of $S'$ $(\subseteq V)$ such that $S'_i \in \mathcal{I}_i$ for all $i \in [k]$, $|S'| > |S|$, and $d_{G(S_1', \dots, S_k')}(s, T) \geq d_{G(S_1, \dots, S_k)}(s, T) + 1$, or a partition $(S_1, \dots , S_k)$ of $S$ if no such $S'$ exists.
\tatsuya{The procedure \textup{\texttt{BlockFlowIndependence}} uses $O(k n + p \log p)$ independence oracle queries.}
\end{lemma}

Cunningham \cite{cunningham1986improved} showed the following lemma; see \cite[Section 5]{cunningham1986improved} for a proof.
We use this lemma for a proof of Lemma \ref{lem:blockflowind}.

\begin{lemma}[\cite{cunningham1986improved}]\label{lem:augmenting_path_monotonicity}
For all vertices $v \in V$ on a shortest augmenting path $P$, the distance from $s$ to $v$ must strictly increase after the augmentation along the path $P$.
\end{lemma}

\begin{proof}[Proof of Lemma \ref{lem:blockflowind}]
We denote by $D_i$ the set of vertices $v \in V$ such that the distance from $s$ to $v$ is $i$ in the compressed exchange graph $G(\Ss)$,
where $(\Ss)$ is the input partition of $S$.
Suppose we obtain a partition $(S_1^{*}, \dots, S_k^{*})$ of $S^*$ by augmenting a partition $(\Ss)$ of $S$ along shortest augmenting paths of length $d_T$.
By Lemma \ref{lem:monotonicity_matroid_partition},
for any shortest path $s, v_1, v_2 \dots, v_{d_T - 1}, t$ of length $d_T$ in the compressed exchange graph $G(S_1^{*}, \dots, S_k^{*})$, 
we have $v_i \in D_i$ for all $i \in [d_T - 1]$.

Now we analyze the number of independence oracle queries used in the procedure \texttt{BlockFlowIndependence}.
By Lemma \ref{lem:getdistanceind}, the number of independence oracle queries used in Line \ref{line:blockflowind1} is \tatsuya{$O(k n + p \log p)$}.
The number of independence oracle queries used in Line \ref{line:blockflowind2} is $O(k |D_{d_T - 1}|)$.
In Line \ref{line:blockflowind3},
the number of \texttt{FindOutEdge} calls that output $\emptyset$ is $O(k \sum_{i = 1}^{d_T - 2} |D_i|)$ 
and the number of \texttt{FindOutEdge} calls that do not output $\emptyset$ is $O(p)$.
Thus, by Lemma \ref{lem:binary_search_matroid_partition_ind}, the number of independence oracle queries used in Line~\ref{line:blockflowind3} is \tatsuya{$O(k n + p \log p)$}.

The procedure \texttt{BlockFlowIndependence} simply repeats to find a shortest path from $s$ to $T$ through the set $L_i$ and to augment the partitionable set and to
remove those vertices from $L_i$ until there is no such path.
Hence, all paths used in \texttt{Update} are shortest augmenting paths.
Then, by Lemma \ref{lem:augmenting_path_matroid_partition_2}, we have $S'_i \in \mathcal{I}_i$ for all $i \in [k]$.
Furthermore, 
there must exist at least one augmenting path of length $d_T$ if $d_T \neq \infty$.
Then, we have $|S'| > |S|$.



Next we show that $d_{G(S_1', \dots, S_k')}(s, T) \geq d_{G(S_1, \dots, S_k)}(s, T) + 1$.
To this end, we prove that any element removed from $L_i$ can no longer be on any shortest augmenting path of length $d_T$.
When a element $a_i \in L_i$ has no outgoing edge into $L_{i + 1}$, the distance from $a_i \in L_i$ to $T$ is greater than $d_T - i$.
Then, by Lemma \ref{lem:monotonicity_matroid_partition}, such $a_i$ can no longer be on any augmenting path of length $d_T$.
Furthermore, 
by Lemma \ref{lem:augmenting_path_monotonicity}, for the vertices $v$ on the augmenting path, the distances from $s$ to $v$ must strictly increase after the augmentation.
Then Lemma \ref{lem:monotonicity_matroid_partition} implies that the element on an augmenting path can not be on any augmenting path of length $d_T$ anymore, which completes the proof.
\end{proof}

Now we provide a proof of Theorem~\ref{thm:matroid_partition_ind_oracle}
by using Lemma \ref{lem:total_bound_matroid_partition}.
In our matroid partition algorithm,
we simply apply \texttt{BlockFlowIndependence} repeatedly until no $(s, T)$-path can be found.

\begin{proof}[Proof of Theorem~\ref{thm:matroid_partition_ind_oracle}]
In our algorithm, we start with $S = \emptyset$ and initialize $S_i = \emptyset$ for all $i \in [k]$.
Then we apply \texttt{BlockFlowIndependence} repeatedly to augment the current partition $(\Ss)$ of $S$
until no $(s, T)$-path can be found in the compressed exchange graph $G(S_1, \dots, S_k)$.

Since each execution of \texttt{BlockFlowIndependence} strictly increases $d_{G(\Ss)}(s, T)$,  we have $d_{G(\Ss)}(s, T) = \Omega(\sqrt{p})$ after $O(\sqrt{p})$ executions of \texttt{BlockFlowIndependence}.
Lemma \ref{lem:total_bound_matroid_partition} implies that,
if $d_{G(\Ss)}(s, T) = \Omega(\sqrt{p})$, then $|S| \geq p - O(\sqrt{p})$.
Then the total number of \texttt{BlockFlowIndependence} executions is $O(\sqrt{p}) + O(\sqrt{p}) = O(\sqrt{p})$ in the entire matroid partition algorithm.
Lemma \ref{lem:blockflowind} implies that one execution of \texttt{BlockFlowIndependence} uses \tatsuya{$O(k n + p\log p)$} independence oracle queries,
which completes the proof.
\end{proof}

In the same way as the matroid intersection algorithm of \cite{chakrabarty2019faster} in the rank oracle model,
we easily obtain the following theorem.

\begin{theorem} \label{thm:matroid_partition_ind_oracle_approx}
For any $\epsilon > 0$,
there is an algorithm that uses \tatsuya{$O\left((k n + p \log p) \epsilon^{-1}\right)$} independence oracle queries and finds a $(1 - \epsilon)$ approximation of the largest partitionable set of $k$ matroids.
\end{theorem} 

\begin{proof}
Similarly to the proof of Theorem~\ref{thm:matroid_partition_ind_oracle}, we start with $S = \emptyset$ and initialize $S_i = \emptyset$ for all $i \in [k]$ and 
apply \texttt{BlockFlowIndependence} repeatedly to augment the current partition $(\Ss)$ of $S$.
The only difference is that we apply \texttt{BlockFlowIndependence} only $\epsilon^{-1}$ times,
which uses \tatsuya{$O\left((k n + p \log p) \epsilon^{-1} \right)$} independence oracle queries.

Each execution of \texttt{BlockFlowIndependence} strictly increases $d_{G(\Ss)}(s, T)$.
Thus, 
after $\epsilon^{-1}$ executions of \texttt{BlockFlowIndependence}, we have $d_{G(\Ss)}(s, T) = \Omega(\epsilon^{-1})$.
Lemma \ref{lem:total_bound_matroid_partition} implies that,
if $d_{G(\Ss)}(s, T) = \Omega(\epsilon^{-1})$, then $|S| \geq p - O(p \epsilon)$, 
which completes the proof.
\end{proof}

In contrast to our implementation provided in Algorithm \ref{alg:blockflowind},
Cunningham \cite{cunningham1986improved} presented another implementation of the procedure \texttt{BlockFlowIndependence}, which does not use the binary search technique.
Recall that, given a partition $(S_1, \dots , S_k)$ of $S$ $(\subseteq V)$ such that $S_i \in \mathcal{I}_i$ for all $i\in [k]$, the procedure \texttt{BlockFlowIndependence} outputs 
a partition $(S'_1, \dots , S'_k)$ of $S'$ $(\subseteq V)$ such that $S'_i \in \mathcal{I}_i$ for all $i \in [k]$, $|S'| > |S|$, and $d_{G(S_1', \dots, S_k')}(s, T) \geq d_{G(S_1, \dots, S_k)}(s, T) + 1$, or a partition $(S_1, \dots , S_k)$ of $S$ if no such $S'$ exists.


We now evaluate the query complexity of Cunningham's implementation.
Since compressed exchange graph has $O(n p)$ edges that are not adjacent to $T$, in a single execution of \texttt{BlockFlowIndependence}, $O(n p)$ independence oracle queries are required for such edges.
The edges adjacent to $T$ are efficiently maintained during the algorithm.
Indeed, we can initially compute such edges with $O(k n)$ independence oracle queries, and update them using $O(n)$ queries after each augmentation.
Thus, we obtain the following remark.

\begin{remark}[Blocking flow algorithm of Cunningham {\cite[Theorem 5.1]{cunningham1986improved}}] \label{rem:cunningham_blocking_flow}
There is an algorithm that performs $d$ executions of the procedure \textup{\texttt{BlockFlowIndependence}} and uses $O(n p d + k n)$ independence oracle queries.
\end{remark}


By Remark~\ref{rem:cunningham_blocking_flow}, Cunningham obtained a matroid partition algorithm that uses $O(n p^{3/2} + k n)$ independence oracle queries, which is better than the algorithm given in Theorem~\ref{thm:matroid_partition_ind_oracle} when $k = \omega(p)$.
Recall that the total number of \texttt{BlockFlowIndependence} executions is $O(\sqrt{p})$ in the entire matroid partition algorithm.
As we will see later in Section \ref{sec:faster_algorithm_for_large_k}, we present an algorithm that is strictly better than Cunningham's algorithm even when $k$ is large.

\subsection{Blocking Flow Algorithm using Rank Oracle} \label{subsec:blockingflow_rank}

In this subsection, we present our matroid partition algorithm using the blocking flow approach in the rank oracle model.
Our objective is to show the following theorem, which implies Theorem \ref{matroid_partition_rank_oracle}.

\begin{theorem} \label{thm:matroid_partition_rank_oracle}
There is an algorithm that uses $O((n + k) \sqrt{p} \log n)$ rank oracle queries and solves the matroid partitioning problem.
\end{theorem}

Note that we can check, for each vertex $u \in S$, whether there exists an edge entering the vertex $u$
by one call of the binary search procedure \texttt{FindInEdge} in the rank oracle model.
This yields that the matroid partitioning problem can be solved by using fewer queries under the rank oracle model than under the independence oracle model.

The rank query complexity of the algorithm given in Theorem~\ref{thm:matroid_partition_rank_oracle} is $\tilde{O}(n^{3/2})$ when $k \leq n$.
As we will see later, the independence query complexity of our matroid partition algorithm given in Theorem \ref{thm:faster_matroid_partition_ind_oracle_2} is $\tO(n^{7/3})$, which is currently the best query complexity under the independence oracle model.

The argument is similar to the algorithm given in Theorem \ref{thm:matroid_partition_ind_oracle}.
However, unlike the algorithm in the independence oracle model, the algorithm in the rank oracle model performs
a breadth first search and an augmentation procedure by the blocking flow approach
by traversing edges in reverse.

We first provide the procedure \texttt{GetDistanceRank} (Algorithm \ref{alg:getdistancerank}) that efficiently finds distances from every vertex to $T$ in the compressed exchange graph.
This algorithm simply runs breadth first search by using the procedure \texttt{FindInEdge}.

\begin{lemma}[Breadth First Search using Rank Oracle]\label{lem:getdistancerank}

Given a partition $(S_1, \dots , S_k)$ of $S$ $(\subseteq V)$ such that $S_i \in \mathcal{I}_i$ for all $i\in [k]$, 
the procedure $\textup{\texttt{GetDistanceRank}}$ (Algorithm \ref{alg:getdistancerank}) outputs $d \in \mathbb{R}^{V \cup \{s\} \cup T}$ such that, for $v \in V \cup \{s\} \cup T$, 
$d_v$ is the distance from $v$ to $T$ in the compressed exchange graph $G(S_1, \dots , S_k)$.
The procedure $\textup{\texttt{GetDistanceRank}}$ uses $O((n + k) \log n)$ rank queries.
\end{lemma}
 

\begin{algorithm}[H]
    \KwInput{a partition $(S_1, \dots , S_k)$ of $S$ $(\subseteq V)$ such that $S_i \in \mathcal{I}_i$ for all $i \in [k]$}
    \KwOutput{$d \in \mathbb{R}^{V \cup \{s\} \cup T}$ such that for $v \in V \cup \{s\} \cup T$, $d_v$ is the distance from $v$ to $T$ in $G(S_1, \dots , S_k)$}
    For $v \in S$ denote by $\pi(v)$ the index such that $v \in S_{\pi(v)}$ \\
    For all $v \in V \cup \{ s\}$ let $d_v \gets \infty$ \\
    $Q \gets \emptyset$ \tcp{$Q :$ queue} 
    $A \gets V$ \\
    \For{$i \in [k]$} {
        $d_{t_i} \gets 0$ \\
        \While{$u$ = \textup{\texttt{FindInEdge}}($M_{i}, S_{i}, \emptyset, A \setminus S_{i}$) satisfies $u \neq \emptyset$} { \label{line:getdistancerank_1}
            $d_u \gets 1$ \\
            $Q \gets Q + u$ \\
            $A \gets A - u$ \\
        }
    }
    \While{$Q \neq \emptyset$}{
        Let $v$ be the element added to $Q$ earliest \\
        $Q \gets Q - v$ \\
        \If{$v \notin S$} {
            \If{$d_s = \infty$} {
                $d_s \gets d_v + 1$ \\
            }
        } \Else {
            \While{$u$ = \textup{\texttt{FindInEdge}}($\M_{\pi(v)}, S_{\pi(v)}, v, A \setminus S_{\pi(v)}$) satisfies $u \neq \emptyset$} { \label{line:getdistancerank_2}
                $d_u \gets d_v + 1$ \\
                $Q \gets Q + u$ \\
                $A \gets A - u$ \\
            }
        }
    }
    \Return {$d$}
    \caption{\texttt{GetDistanceRank}}\label{alg:getdistancerank}
\end{algorithm}

\begin{proof}
The procedure \texttt{GetDistanceRank} simply performs a breadth first search in the compressed exchange graph $G(S_1, \dots , S_k)$.
Thus, the procedure \texttt{GetDistanceRank} correctly computes distances to $T$ in $G(S_1, \dots , S_k)$.
Note that each vertex $v \in V$ is added to $Q$ at most once and removed from $A$ at most once.
Hence, 
the number of \texttt{FindInEdge} calls that do not output $\emptyset$ is $O(n)$, and the number of \texttt{FindInEdge} calls that output $\emptyset$ is $O(n + k)$.
Thus, by Lemma \ref{lem:binary_search_matroid_partition_rank}, the number of rank oracle queries used in Lines \ref{line:getdistancerank_1} and \ref{line:getdistancerank_2} is $O((n + k) \log n)$, which completes the proof.
\end{proof}

Next, in a similar way to the procedure \texttt{BlockFlowIndependence} in the independence oracle model,
we provide our augmentation subroutine for a fast matroid partition algorithm in the rank oracle model.
The implementation is described as \texttt{BlockFlowRank} (Algorithm \ref{alg:blockflowrank}).
In the procedure \texttt{BlockFlowRank}, we 
denote by $I$ the set of indices $j$ such that the sink vertex $t_j \in T$ is not removed.

\begin{algorithm}[H]
    \KwInput{a partition $(S_1, \dots , S_k)$ of $S$ $(\subseteq V)$ such that $S_i \in \mathcal{I}_i$ for all $i \in [k]$}
    \KwOutput{a partition $(S'_1, \dots , S'_k)$ of $S'$ $(\subseteq V)$ such that $S'_i \in \mathcal{I}_i$ for all $i \in [k]$, $|S'| > |S|$, and $d_{G(S_1', \dots, S_k')}(s, T) \geq d_{G(S_1, \dots, S_k)}(s, T) + 1$, or a partition $(S_1, \dots , S_k)$ of $S$ if no such $S'$ exists.}
    $d \gets \text{\texttt{GetDistanceRank}}(S_1, \dots , S_k)$ \label{line:blockflowrank1} \\
    \lIf{$d_s = \infty$} {
        \Return $(S_1, \dots , S_k)$
    }
    $I \gets [k]$ \\
    For all $i \in [d_s - 1]$ let $L_i \gets \{ v \in V \mid d_v = i \}$ \\
    For $v \in S$ denote by $\pi(v)$ the index such that $v \in S_{\pi(v)}$ \\
    \While{$L_i \neq \emptyset$ for all $i \in [d_s - 1]$ and $I \neq \emptyset$} {
        $l \gets 0$ , $a_0 \gets \emptyset$ , $a_{d_s - 1} \gets \emptyset$ \\
        \While{$l < d_{s} - 1$} {
            $a_{l + 1} \gets \emptyset$ \\
            \If{$l = 0$} {
                \While{$I \neq \emptyset$} {
                    Pick arbitrary $j \in I$ \\
                    $a_1 \gets $ \texttt{FindInEdge}($\M_j, S_j, \emptyset, L_1 \setminus S_j$) \label{line:blockflowrank2} \\
                    \If{$a_1 \neq \emptyset$} {
                        $a_0 \gets t_j$ \\
                        \Break
                    } \Else {
                        $I \gets I - j$ \\
                    }
                }
                \lIf{$a_1 = \emptyset$} {
                    \Break
                }
            } \Else {
                $a_{l + 1} \gets $ \texttt{FindInEdge}($\M_{\pi(a_l)}, S_{\pi(a_l)}, a_l, L_{l + 1} \setminus S_{\pi(a_l)}$) \label{line:blockflowrank3} \\
            }

            \If{$a_{l + 1} = \emptyset$} {
                $L_l \gets L_l - a_l$ \\
                \lIf{$L_l = \emptyset$} {
                    \Break
                }
                $l \gets l - 1$ \\
            } \Else {
                $l \gets l + 1$ \\
            }
        }

        \If{$a_{d_s - 1} \neq \emptyset$} {
            Denote by $P$ the path $s, a_{d_s - 1}, \ldots , a_{1}, a_{0}$ \\
            $(S_1, \dots , S_k) \gets \text{\texttt{Update}}((S_1, \dots , S_k), P)$ \\
            For all $i \in [d_s - 1]$ let $L_i \gets L_i - a_i$ \\
        }
    }
    \Return $(S_1, \dots , S_k)$
    \caption{\texttt{BlockFlowRank}}\label{alg:blockflowrank}
\end{algorithm}


\begin{lemma}[Blocking Flow using Rank Oracle] \label{lem:blockflowrank}
Given a partition $(S_1, \dots , S_k)$ of $S$ $(\subseteq V)$ such that $S_i \in \mathcal{I}_i$ for all $i\in [k]$, 
the procedure \textup{\texttt{BlockFlowRank}} (Algorithm \ref{alg:blockflowrank}) outputs 
a partition $(S'_1, \dots , S'_k)$ of $S'$ $(\subseteq V)$ such that $S'_i \in \mathcal{I}_i$ for all $i \in [k]$, $|S'| > |S|$, and $d_{G(S_1', \dots, S_k')}(s, T) \geq d_{G(S_1, \dots, S_k)}(s, T) + 1$, or a partition $(S_1, \dots , S_k)$ of $S$ if no such $S'$ exists.
The procedure \textup{\texttt{BlockFlowRank}} uses $O((n + k) \log n)$ rank oracle queries.
\end{lemma}

\begin{proof}
We denote by $D_i$ the set of vertices $v \in V$ such that the distance from $v$ to sink vertices $T$ is $i$ in the compressed exchange graph $G(\Ss)$,
where partition $(\Ss)$ is the input partition of $S$.
We first analyze the number of rank oracle queries used in the procedure \texttt{BlockFlowRank}.
By Lemma \ref{lem:getdistancerank}, the number of rank oracle queries used in Line \ref{line:blockflowrank1} is $O((n + k) \log n)$.
In Line \ref{line:blockflowrank2},
the number of \texttt{FindInEdge} calls that output $\emptyset$ is $O(k)$, and
the number of \texttt{FindInEdge} calls that do not output $\emptyset$ is $O(|D_1|)$.
Furthermore,
in Line \ref{line:blockflowrank3},
the number of \texttt{FindInEdge} calls that output $\emptyset$ is $O(p)$, and
the number of \texttt{FindInEdge} calls that do not output $\emptyset$ is $O(\sum_{i = 2}^{d_s - 1} |D_i|)$.
Hence, by Lemma \ref{lem:binary_search_matroid_partition_rank}, the number of rank oracle queries used in Lines \ref{line:blockflowrank2} and \ref{line:blockflowrank3} is $O((n + k) \log n)$.

By the same argument as the proof of Lemma \ref{lem:blockflowrank},
we have $S'_i \in \mathcal{I}_i$ for all $i \in [k]$ and $|S'| > |S|$.

Next, we show that $d_{G(S_1', \dots, S_k')}(s, T) \geq d_{G(S_1, \dots, S_k)}(s, T) + 1$.
To this end, we prove that any element removed from $L_i$ can no longer be on any shortest augmenting path of length $d_T$.
When an element $a_i \in L_i$ has no incoming edge from $L_{i + 1}$, the distance from $s$ to $a_i$ is larger than $d_T - i$.
Then, by Lemma \ref{lem:monotonicity_matroid_partition}, such $a_i$ can no longer be on any augmenting path of length $d_T$.
The vertex $t_j \in T$ such that the element $j$ is removed from $I$ also can not be on an augmenting path of length $d_T$ anymore.
By Lemma \ref{lem:augmenting_path_monotonicity}, for every vertex $v$ on an augmenting path, the distance from $s$ to $v$ must strictly increase after the augmentation.
Thus Lemma \ref{lem:monotonicity_matroid_partition} implies that the element on an augmenting path can not be on an augmenting path of length $d_T$ anymore, which completes the proof.
\end{proof}

Now we give a proof of Theorem~\ref{thm:matroid_partition_rank_oracle}.

\begin{proof}[Proof of Theorem~\ref{thm:matroid_partition_rank_oracle}]
In the same way as the proof of Theorem \ref{thm:matroid_partition_ind_oracle},
we start with $S = \emptyset$ and initialize $S_i = \emptyset$ for all $i \in [k]$.
Then we apply \texttt{BlockFlowRank} repeatedly to augment the current partition $(\Ss)$ of $S$
until no $(s, T)$-path can be found in the compressed exchange graph $G(S_1, \dots, S_k)$.

By the same argument as Theorem \ref{thm:matroid_partition_ind_oracle}, 
the total number of \texttt{BlockFlowRank} executions is $O(\sqrt{p})$ in the entire matroid partition algorithm.
By Lemma \ref{lem:blockflowrank}, the proof is complete.
\end{proof}

We obtain the following theorem by the same argument as Theorem \ref{thm:matroid_partition_ind_oracle_approx}. 

\begin{theorem}\label{thm:matroid_partition_rank_oracle_approx}
For any $\epsilon > 0$,
there is an algorithm that uses $O((n + k) \epsilon^{-1} \log n)$ rank oracle queries and finds a $(1 - \epsilon)$ approximation to the largest partitionable set of $k$ matroids.
\end{theorem} 

\begin{proof}
In the same way as the proof of Theorem \ref{thm:matroid_partition_ind_oracle_approx}, 
we start with $S = \emptyset$ and initialize $S_i = \emptyset$ for all $i \in [k]$ and
apply \texttt{BlockFlowRank} only $\epsilon^{-1}$ times.
By Lemma \ref{lem:blockflowrank}, the entire algorithm uses $O((n + k) \epsilon^{-1} \log n)$ rank oracle queries.
\end{proof}

\section{Faster Algorithm for Large \texorpdfstring{$k$}{TEXT}} \label{sec:faster_algorithm_for_large_k}

\tatsuya{
In this section, we present an algorithm that uses $o(\min\{k, p\} n \sqrt{p})$ independence oracle queries
when $k$ is large.
In Subsection \ref{subsec:blockingflow_ind}, we have presented the algorithm (Algorithm \ref{alg:getdistanceind}), 
which runs a breadth first search in the compressed exchange graph and uses $\tilde{O}(k n)$ independence oracle queries. 
Recall that the compressed exchange graph has $O(n p)$ edges, and then we can construct the compressed exchange graph by using $O(n p)$ independence oracle queries.
Thus, there is an algorithm that uses $\tO(\min\{ k n, n p \})$ independence oracle queries and runs a breadth first search in the compressed exchange graph.
In the evaluation of the independence query complexity of the matroid partition algorithm by the \emph{blocking flow} approach given in Section \ref{sec:blockingflow},
a key observation is
that the number of different lengths of shortest augmenting paths during the algorithm is $O(\sqrt{p})$.
For now, it is not clear whether we can obtain a matroid partition algorithm that runs a breadth first search $o(\sqrt{p})$ times.
Then the blocking flow approaches are now stuck at $\Omega(\min\{k, p\} n \sqrt{p})$ independence oracle queries.
}

To overcome this barrier and improve upon the algorithm that uses $\tilde{O}(\min\{k, p\} n \sqrt{p})$ independence oracle queries given by Cunningham~\cite{cunningham1986improved} and Theorem \ref{thm:matroid_partition_ind_oracle},
we introduce a new approach called \emph{edge recycling augmentation}, 
which can perform breadth first searches with fewer total independence oracle queries.
Our new approach can be attained through new ideas: 
an efficient utilization of the binary search procedure \texttt{FindOutEdge} and
a careful analysis of the number of independence oracle queries by using Lemma \ref{lem:total_bound_matroid_partition}.
By combining an algorithm by the blocking flow approach and an algorithm by the edge recycling augmentation approach,
we obtain the following theorem, which implies Theorem \ref{faster_matroid_partition_ind_oracle}.

\begin{theorem}\label{thm:faster_matroid_partition_ind_oracle_2}
There is an algorithm that uses $O(k'^{1 / 3} n p \log p + k n)$ independence oracle queries and solves the matroid partitioning problem, where $k' = \min \{k, p\}$.
\end{theorem}

\tatsuya{This theorem implies that we obtain a matroid partition algorithm that uses $o(\min\{k, p\} n \sqrt{p})$ independence oracle queries
when $k = \omega(p^{3/4})$.
We note that this algorithm requires $O(k'^{2/3} n p)$ time complexity other than independence oracle queries.
}

In Section \ref{subsec:faster_augumentation}, we present our new approach edge recycling augmentation, and in Section \ref{subsec:combining},  we present our faster matroid partition algorithm for large $k$ and give a proof of Theorem \ref{thm:faster_matroid_partition_ind_oracle_2}.

\subsection{Edge Recycling Augmentation} \label{subsec:faster_augumentation}

In order to select appropriate parameters for our algorithm, 
we have to determine the value of $p$.
However, the size $p$ of a largest partitionable set is unknown before running the algorithm.
Instead of using the exact value of $p$, we use a $\frac{1}{2}$-approximation $\bar{p}$ for $p$ (that is $\bar{p} \leq p \leq 2 \bar{p}$), which can be computed using $O(k n)$ independence oracle queries.
It is well known that a $\frac{1}{2}$-approximate solution for the matroid intersection problem can be found by the following simple greedy algorithm; see \cite[Proposition 13.26]{korte2011combinatorial}.
We begin with an empty set. For each element in the ground set, we check whether adding it to the set would result in a common independent set. If it does, we add it to the set.
Finally, we obtain a maximal common independent set.
We convert this algorithm into the following $\frac{1}{2}$-approximation algorithm (Algorithm \ref{alg:approximation_matroid_partition}) for the matroid partitioning problem by utilizing 
the reduction from matroid partition to the intersection of two matroids $\hat{\M}' = (\hat{V}, \hat{\mathcal{I}}')$ and $\hat{\M}'' =(\hat{V}, \hat{\mathcal{I}}'')$
given in subsection \ref{subsec:compressed_exchange_graph}.

\begin{algorithm}[H]
    For all $i \in [k]$ let $S_i \gets \emptyset$ \\
    \For{$i \leftarrow 1$ \KwTo $k$} {
        \For{$v \in V \setminus \left( \bigcup_{j = 1}^{i - 1} S_j \right) $} {
            \If{$S_i + v \in \mathcal{I}_i$} {
                $S_i \gets S_i + v$ \\
            }
        }
    }
    \Return{ $\bar{p} = |\bigcup_{i = 1}^{k} S_i|$ }
    \caption{\texttt{$\frac{1}{2}$-ApproximationMatroidPartition}}\label{alg:approximation_matroid_partition}
\end{algorithm}

Now we present our new approach \emph{Edge Recycling Augmentation}.
Our new approach edge recycling augmentation is applied in each phase of the algorithm.
One phase of edge recycling augmentation is described as \texttt{EdgeRecyclingAugmentation} (Algorithm \ref{alg:FasterFindingPaths}).

In \texttt{EdgeRecyclingAugmentation}, we first compute the edges $E^* ( \subseteq V \times S)$ in the compressed exchange graph $G(\Ss)$,
which uses $O(n p)$ independence oracle queries.
Note that the compressed exchange graph may be changed by augmentations, 
that is, augmentations may add or delete several edges in the compressed exchange graph, 
and so, taking one augmenting path may destroy the set $E^*$ of the edges.
However, we notice that we can \emph{recycle} some part of the edge set $E^*$ after the augmentations,
which is peculiar to the matroid partition.


In \texttt{EdgeRecyclingAugmentation}, we simply repeat to run a breadth first search and then to augment the partitionable set.
Unlike \texttt{GetDistanceIndependence} (Algorithm \ref{alg:getdistanceind}) in section \ref{subsec:blockingflow_ind}, our BFS \emph{recycles} the precomputed edge set $E^*$.
In one phase, we keep a set $J$ of all indices $i$ such that $S_i$ was updated by the augmentations.
Our crucial observation is that no edges, in the compressed exchange graph, entering a vertex in $S_i$ are changed by the augmentations
unless augmenting paths contain a vertex in $S_i \cup \{ t_i \}$.
In contrast to \texttt{GetDistanceIndependence} that uses the binary search procedure \texttt{FindOutEdge} for all indices $i \in [k]$,
our new BFS procedure uses \texttt{FindOutEdge} only for indices $i \in J$.
We can use $E^*$ to search edges entering a vertex in $S_i$ with $i \notin J$.
Then, the BFS based on the ideas described above can be implemented as \texttt{EdgeRecyclingBFS} (Algorithm \ref{alg:FasterBFS}).


We also provide a new significant analysis of the number of independence oracle queries in entire our matroid partition algorithm.
In \texttt{EdgeRecyclingAugmentation},
we repeat to run the breadth first search \texttt{EdgeRecyclingBFS} so that the total calls of the binary search procedure is $O(n p)$.
Then, the number of independence oracle queries used by \texttt{EdgeRecyclingBFS} in one call of \texttt{EdgeRecyclingAugmentation} is almost equal to the one used by the precomputation of $E^*$.
Hence, one call of \texttt{EdgeRecyclingAugmentation} uses $\tO(n p)$ independence oracle queries.
The number of calls of \texttt{EdgeRecyclingBFS} in \texttt{EdgeRecyclingAugmentation} depends on how many edges can not be recycled.
Thus, to determine the number of calls of \texttt{EdgeRecyclingBFS}, we use the value $sum$ in the implementation of \texttt{EdgeRecyclingAugmentation} (Algorithm \ref{alg:FasterFindingPaths}).
In the entire matroid partition algorithm, we apply \texttt{EdgeRecyclingAugmentation} repeatedly.
Then,
we can obtain a matroid partition algorithm that uses $\tO(np^{3/2} + k n)$ independence oracle queries.
\tatsuya{
Furthermore, by combining this with the blocking flow approach, 
the number of total calls of \texttt{EdgeRecyclingAugmentation} in the entire matroid partition algorithm can be $O(\min\{k^{1/3}, p^{1/3} \})$.
This leads to obtain a matroid partition algorithm that uses $\tO(\min\{k^{1/3}, p^{1/3} \} n p + k n)$ independence oracle queries. 
}
This analysis differs significantly from that of existing faster matroid intersection algorithms.




For $i \in [k]$, let $F_i (\subseteq V )$ denote the set of vertices adjacent to $t_i \in T$.
We first compute the set $F_i$ for all $i \in [k]$ using $O(k n)$ independence oracle queries.
Note that, after one augmentation, they can be updated using only $O(n)$ independence oracle queries.

In the following two lemmas, we show the correctness and the independence query complexity of the procedure \texttt{EdgeRecyclingAugmentation} (Algorithm~\ref{alg:FasterFindingPaths}).

\begin{lemma}\label{lem:faster_matroid_partition_correctness}
Given a partition $(S_1, \dots, S_k)$ of $S (\subseteq V)$ such that $S_i \in \mathcal{I}_i$ for all $i \in [k]$,
the procedure \textup{\texttt{EdgeRecyclingAugmentation}} (Algorithm~\ref{alg:FasterFindingPaths}) outputs 
a partition $(S'_1, \dots, S'_k)$ of $S' (\subseteq V)$ such that $S'_i \in \mathcal{I}_i$ for all $i \in [k]$ and $|S'| \geq |S|$.
\end{lemma}
\begin{proof}
To prove the correctness of \texttt{EdgeRecyclingAugmentation}, we prove the following invariants at the beginning of any iteration of the while loop.
\begin{itemize}
 \item[(i)]  For all $i \in [k] \setminus J$ and all $(v, u) \in V \times S_i$, we have $(v, u) \in E^*$ if and only if $S_i + v - u \in \mathcal{I}_i$ and $S_i + v \notin \I_i$.
 \item[(ii)] For all $i \in [k]$ and all $v \in V$, we have $v \in F_i$ if and only if $S_i + v \in \mathcal{I}_i$ and $v \notin S_i$.
 \item[(iii)] For all $i \in [k]$, we have $S_i \in \mathcal{I}_i$.
\end{itemize}


The invariant is true before the execution of \texttt{EdgeRecyclingAugmentation}.
Now, assume that the invariants (i)--(iii) hold true at the beginning of an iteration of the while loop.
Let a partition $(S_1^{{\rm old}}, \dots, S_k^{{\rm old}})$ of $S^{{\rm old}}$ be the partition before the execution of Line \ref{line:fasteraugmentations5} and
a partition $(S_1^{{\rm new}}, \dots, S_k^{{\rm new}})$ of $S^{{\rm new}}$ be the partition after the execution of Line \ref{line:fasteraugmentations5}.
For all $i \in [k] \setminus J$, we have $S_i^{{\rm old}} = S_i^{{\rm new}}$.
Then, invariant (i) remains true. 
For all $i \in [k] \setminus \{ j \}$, we have $|S_i^{{\rm old}}| = |S_i^{{\rm new}}|$.
It is well known that $S_i^{{\rm new}}$ has the same span as $S_i^{{\rm old}}$ when $|S_i^{{\rm old}}| = |S_i^{{\rm new}}|$; see \cite[Lemma 4.4]{cunningham1986improved}.
Then, invariant (ii) remains true. 
The procedure \texttt{EdgeRecyclingBFS} simply finds a BFS-tree rooted at $s$ by a breadth first search.
Thus, if the invariants (i)--(iii) are true, then the procedure \texttt{EdgeRecyclingBFS} correctly computes BFS-tree rooted at $s$.
Then, the path $P$ that \texttt{EdgeRecyclingBFS} outputs in Line~\ref{line:fasteraugmentations6} is a shortest augmenting path.
Hence, by Lemma \ref{lem:augmenting_path_matroid_partition}, invariant (iii) remains true.
\end{proof}

\begin{lemma}\label{lem:faster_matroid_for_k_inequality_1}
The procedure \textup{\texttt{EdgeRecyclingAugmentation}} (Algorithm~\ref{alg:FasterFindingPaths}) uses $O(n p \log p)$ independence oracle queries.
\end{lemma}
\begin{proof}
The number of independence oracle queries used in Line \ref{line:fasteraugmentations1} is $O(n p)$.
Furthermore, the number of independence oracle queries used in Line \ref{line:fasteraugmentations4} is $O(n p)$,
because the number of iterations of the while loop is bounded by $p$.
It remains to show that the number of \texttt{FindOutEdge} calls in the entire procedure \texttt{EdgeRecyclingAugmentation} is $O(n p)$.

In the procedure \texttt{EdgeRecyclingBFS}($(S_1, \dots, S_k), E^*, J, \bigcup_{i = 1}^{k}F_i$), each vertex $v \in V$ is added to $Q$ at most once and each vertex $v \in S$ is removed from $B_{\pi(v)}$ at most once,
where $\pi(v)$ is the index such that $v \in S_{\pi(v)}$.
This means that the number of \texttt{FindOutEdge} calls that do not output $\emptyset$ is bounded by $p$, and the number of \texttt{FindOutEdge} calls that output $\emptyset$ is bounded by $n \cdot |J|$.
Then, the number of \texttt{FindOutEdge} calls in the procedure \texttt{EdgeRecyclingBFS} is $O(p + n \cdot |J|)$.

Suppose that \texttt{EdgeRecyclingBFS} is called for $J = J_1, J_2, \dots, J_c$ in the procedure \texttt{EdgeRecyclingAugmentation}.
Obivously, $c \leq 2\bar{p} = O(p)$.
Furthermore, by the condition of the while loop in \texttt{EdgeRecyclingAugmentation}, $\displaystyle \sum_{i = 1}^{c}|J_i| = O(p) $.
Thus, the number of \texttt{FindOutEdge} calls in the entire procedure \texttt{EdgeRecyclingAugmentation} is $\displaystyle O \left( \sum_{i = 1}^{c}(p + n \cdot |J_i|) \right)$, which is $O(n p)$.
Hence, by Lemma \ref{lem:binary_search_matroid_partition_ind}, the number of independence oracle queries by \texttt{FindOutEdge} in the entire procedure \texttt{EdgeRecyclingAugmentation} is $O(n p \log p)$, which completes the proof.
\end{proof}

\begin{algorithm}[H]
    \KwInput{a partition $(S_1, \dots , S_k)$ of $S$ $(\subseteq V)$ such that $S_i \in \mathcal{I}_i$ for all $i \in [k]$,
    a set $E \subseteq V \times S$,
    a set $J \subseteq [k]$,
    a set $F = \{ v \in V \mid \exists i \in [k], v \notin S_i, S_i + v \in \mathcal{I}_i \}$.}
    \KwOutput{an augmenting $(s, T)$-path in $G(S_1, \dots , S_k)$ if one exists.}
    $Q \gets \{ v \in V \setminus S \}$ \tcp{$Q :$ queue} 
    $B_i \gets S_i$ for all $i \in [k]$ \\
    \While{$Q \neq \emptyset$}{
        Let $v$ be the element added to $Q$ earliest \\
        $Q \gets Q - v$ \\
        \If{$v \in F$} {
            \Return the shortest augmenting path in the BFS-tree
        }
        \For{$i \in J$}{
            \While{$u$ = $\text{\textup{\texttt{FindOutEdge}}}$($\M_i, S_i, v, B_i$) satisfies $u \neq \emptyset$} {
                $Q \gets Q + u$ \\
                $B_i \gets B_i - u$
            }
        }

        \For{$i \in [k] \setminus J$} {
            \For{$u \in B_i$ such that $(v, u) \in E$} {
                $Q \gets Q + u$ \\
                $B_i \gets B_i - u$ \\
            }
        }
    }
    \Return NO PATH EXISTS
    \caption{\texttt{EdgeRecyclingBFS}}\label{alg:FasterBFS}
\end{algorithm}

\begin{algorithm}[H]
    \KwInput{a partition $(S_1, \dots , S_k)$ of $S$ $(\subseteq V)$ such that $S_i \in \mathcal{I}_i$ for all $i \in [k]$,
    sets $F_i = \{ v \in V \setminus S_i \mid S_i + v \in \mathcal{I}_i \}$ for all $i \in [k]$}
    \KwOutput{a partition $(S'_1, \dots , S'_k)$ of $S'$ $(\subseteq V)$ such that $S'_i \in \mathcal{I}_i$ for all $i \in [k]$ and $|S'| \geq |S|$.}
    $sum \gets 0$ \\
    $J \gets \emptyset$ \\
    $E^* \gets \{(v, u) \in V \times S \mid \exists i \in [k], u \in S_i, S_i + v \notin \mathcal{I}_i, S_i + v - u \in \mathcal{I}_i \}$ \\ \label{line:fasteraugmentations1}
    \While{$sum < 2 \bar{p}$}{
        $P \gets $ \texttt{EdgeRecyclingBFS}($(S_1, \dots, S_k), E^*, J, \bigcup_{i = 1}^{k}F_i$) \label{line:fasteraugmentations6} \\
        \If{$P$ = \rm{NO PATH EXISTS}} {
            \Break \label{line:fasteraugumentation_break}
        }
        For $v \in S$ denote by $\pi(v)$ the index such that $v \in S_{\pi(v)}$ \\
        Denote by $V(P) = \{s, v_1, \dots , v_{l - 1}, t_j\}$ the vertices in the path $P$ \\
        \For{ $i \leftarrow 2$ \KwTo $l - 1$ } { \label{line:fasteraugmentations2}
            $J \gets J + \pi(v_i)$ \\
        }
        $J \gets J + j$  \label{line:fasteraugmentations3} \\
        $(S_1, \dots , S_k) \gets \text{\texttt{Update}}((S_1, \dots , S_k), P)$ \label{line:fasteraugmentations5} \\
        $F_j \gets \{v \in V \mid v \notin S_j$ and $S_j + v \in \mathcal{I}_j \}$ \label{line:fasteraugmentations4} \\
        $sum \gets sum + |J|$ \\
    }
    \Return{ $(S_1, \dots , S_k)$}
    \caption{\texttt{EdgeRecyclingAugmentation}}\label{alg:FasterFindingPaths}
\end{algorithm}

At this point, we can obtain a matroid partition algorithm that uses $O(n p^{3/2} \log p + k n)$ independence oracle queries.
In the algorithm, we first compute $F_i = \{ v \in V \setminus S_i \mid S_i + v \in \mathcal{I}_i \}$ for all $i \in [k]$.
Next, we apply \texttt{EdgeRecyclingAugmentation} repeatedly
to augment the current partition $(S_1, \dots, S_k)$ of $S$ 
until no $(s, T)$-path can be found in the compressed exchange graph $G(S_1, \dots, S_k)$.
As we will show later in Lemma \ref{lem:analysis_matroid_parition_for_large_k_2},
the number of independence oracle queries in this algorithm is $O(n p^{3/2} \log p + k n)$.
In the next subsection,
we improve this by combining the algorithm by the blocking flow approach and the algorithm by the edge recycling augmentation approach.

\subsection{Going Faster for Large \texorpdfstring{$k$}{TEXT} by Combining Blocking Flow and Edge Recycling Augmentation} \label{subsec:combining}

We have already presented two algorithms to solve the matroid partitioning problem in the independence oracle model.
We combine the algorithm by the \emph{blocking flow} approach and the one by the \emph{edge recycling augmentation} approach.
When the distance from $s$ to $T$ in the compressed exchange graph is small, we use the blocking flow approach.
On the other hand, when the distance from $s$ to $T$ in the compressed exchange graph is large, we use the edge recycling augmentation approach.
The implementation is described as Algorithm \ref{alg:FasterMatroidPartitionForLaergek}.
Then we obtain a matroid partitioning algorithm that uses \tatsuya{$o(\min \{k, p \} n \sqrt{p})$} independence oracle queries
when $k = \omega(p^{3/4})$.
This improves upon both Cunningham's algorithm and the algorithm given in Theorem \ref{thm:matroid_partition_ind_oracle}, both of which use only the blocking flow approach.


\begin{algorithm}[H]
    Compute a $\frac{1}{2}$-approximation $\bar{p}$ for $p$ by running \texttt{$\frac{1}{2}$-ApproximationMatroidPartition} (Algorithm \ref{alg:approximation_matroid_partition}) and determine the value of $d$. \label{line:faster_matorid_partition_4} \\
    For all $i \in [k]$ let $S_i \gets \emptyset$ \\
    Apply \texttt{BlockFlowIndependence} (Algorithm \ref{alg:blockflowind}) repeatedly
    to augment the current partition $(S_1, \dots, S_k)$ of $S$ 
    until the distance from $s$ to $T$ in the compressed exchange graph $G(S_1, \dots, S_k)$ is at least $d$. 
    (We adopt the implementation mentioned in Remark \ref{rem:cunningham_blocking_flow} if it is better than Algorithm \ref{alg:blockflowind}.) \label{line:faster_matorid_partition_1} \\
    For all $i \in [k]$ let $F_i \gets \{ v \in V \setminus S_i \mid S_i + v \in \mathcal{I}_i \}$ \label{line:faster_matorid_partition_3} \\
    Apply \texttt{EdgeRecyclingAugmentation} (Algorithm~\ref{alg:FasterFindingPaths}) repeatedly
    to augment the current partition $(S_1, \dots, S_k)$ of $S$ and to update $F_j$ with $j \in [k]$ 
    until no $(s, T)$-path can be found in the compressed exchange graph $G(S_1, \dots, S_k)$. \label{line:faster_matorid_partition_2}  \\
    \caption{Faster Matroid Partition Algorithm for Large $k$}\label{alg:FasterMatroidPartitionForLaergek}
\end{algorithm}

The algorithm is parametrized by an integer $d$, which we set in the end.
To analyze the independence query complexity of Algorithm \ref{alg:FasterMatroidPartitionForLaergek}, 
we first show that Line \ref{line:faster_matorid_partition_1} uses \tatsuya{$\tilde{O}(\min\{k n d , p n d + k n\})$} independence oracle queries and
Line \ref{line:faster_matorid_partition_2} uses $\displaystyle \tilde{O}\left(\frac{p^{3 / 2} n}{d^{1 / 2}} \right)$ independence oracle queries.

\begin{lemma}\label{lem:analysis_matroid_parition_for_large_k_1}
Line \ref{line:faster_matorid_partition_1} of Algorithm \ref{alg:FasterMatroidPartitionForLaergek} uses 
\tatsuya{$O(\min\{k n d \log p , p n d + k n\})$}
independence oracle queries.
\end{lemma}
\begin{proof}


Lemma \ref{lem:blockflowind} implies that the distance from $s$ to $T$ in the compressed exchange graph increases by at least 1 after the execution of \texttt{BlockFlowIndependence}.
Then, the number of calls of \texttt{BlockFlowIndependence} is bounded by $d$.
Since Lemma \ref{lem:blockflowind} implies that the number of independence oracle queries in one call of \texttt{BlockFlowIndependence} is $O(k n \log p)$, Line \ref{line:faster_matorid_partition_1} of Algorithm \ref{alg:FasterMatroidPartitionForLaergek} can be implemented as an algorithm that uses $O(k n d \log p)$ independence oracle queries.
Furthermore, by Remark~\ref{rem:cunningham_blocking_flow}, Line \ref{line:faster_matorid_partition_1} of Algorithm \ref{alg:FasterMatroidPartitionForLaergek} can also be implemented as an algorithm that uses $O(p n d + k n)$ independence oracle queries, which completes the proof.
\end{proof}

\begin{lemma}\label{lem:analysis_matroid_parition_for_large_k_2}
Line \ref{line:faster_matorid_partition_2} of Algorithm \ref{alg:FasterMatroidPartitionForLaergek} uses $\displaystyle O\left(\frac{p^{3 / 2} n}{d^{1 / 2}} \log p\right)$ independence oracle queries.
\end{lemma}
\begin{proof}

Let $m$ denote the number of calls of \texttt{EdgeRecyclingAugmentation} in Line \ref{line:faster_matorid_partition_2} of Algorithm~\ref{alg:FasterMatroidPartitionForLaergek}.
By Lemma \ref{lem:faster_matroid_for_k_inequality_1}, we only have to show that $\displaystyle m = O\left(\sqrt{\frac{p}{d}}\right)$.
For $i \in [m]$, let $c_i$ denote the number of augmenting paths found in the $i$-th call of \texttt{EdgeRecyclingAugmentation}.
For $i \in [m - 1] \cup \{ 0 \}$, we write $\displaystyle s_i = \sum_{j = i + 1}^{m} c_j$.

We first show the following two claims.
\begin{claim}\label{claim:1}
\tatsuya{
For all $i \in [m - 1]$, we have $\sqrt{s_i} = O(c_i)$.
}
\end{claim}
\begin{proof}
Let $i \in [m - 1]$.
We denote by $l_i$ the length of the augmenting path found in the last \texttt{EdgeRecyclingBFS} in the $i$-th call of \texttt{EdgeRecyclingAugmentation}.
Lemma \ref{lem:total_bound_matroid_partition} implies that 
$\displaystyle s_i = O \left(\frac{p}{l_i}\right) $.
We note that, by Lemma \ref{lem:monotonicity_matroid_partition}, the length of shortest augmenting paths never decreases as the partitionable set size increases.

In the $i$-th call of \texttt{EdgeRecyclingAugmentation},
the sum of the sizes of $J$ is upper bounded by $c_i^2 \cdot l_i$, because the size of $J$ is upper bounded by $c_i \cdot l_i$.
Furthermore, by the condition of the while loop in \texttt{EdgeRecyclingAugmentation},
the sum of the sizes of $J$ is at least $2 \bar{p} ( \geq p)$.
Thus, we obtain $c_i^2 \cdot l_i \geq p$, and then we have $\displaystyle c_i^2 \geq \frac{p}{l_i}$.
Since $\displaystyle s_i = O\left(\frac{p}{l_i}\right) $,
we have $\sqrt{s_i} = O(c_i)$, which is the desired conclusion.
\end{proof}

\tatsuya{
By Claim~\ref{claim:1}, 
there is a positive constant $C$ such that $c_i \geq C \sqrt{s_i}$ for all $i \in [m - 1]$.
}

\begin{claim}\label{claim:2}
For all $i \in [m - 1]$, we have
$\displaystyle \frac{C}{\sqrt{1 + C}} \leq \int_{s_i}^{s_{i - 1}} \frac{dx}{\sqrt{x}}$.
\end{claim}
\begin{proof}
Let $i \in [m - 1]$.
Since $c_i \geq C \sqrt{s_i}$,
we obtain 
\begin{align*}
\int_{s_i}^{s_{i - 1}} \frac{dx}{\sqrt{x}} = 
\int_{s_i}^{s_i + c_i} \frac{dx}{\sqrt{x}} \geq &
\int_{s_i}^{s_i + C \sqrt{s_i}} \frac{dx}{\sqrt{x}} 
\geq \int_{s_i}^{s_i + C \sqrt{s_i}} \frac{dx}{\sqrt{s_i + C \sqrt{s_i}}} \\
= & \frac{C\sqrt{s_i}}{\sqrt{s_i + C \sqrt{s_i}}}
= \frac{C}{\sqrt{1 + C \frac{1}{\sqrt{s_i}}}} \geq \frac{C}{\sqrt{1 + C}},
\end{align*}
which completes the proof.
\end{proof}

By Claim \ref{claim:2},
$ \displaystyle
m - 1 = \sum_{i = 1}^{m - 1} 1 \leq \frac{\sqrt{1 + C}}{C} \sum_{i = 1}^{m - 1} \int_{s_{i}}^{s_{i - 1}} \frac{dx}{\sqrt{x}} = O\left( \int_{s_{m - 1}}^{s_{0}} \frac{dx}{\sqrt{x}} \right) = O\left( \sqrt{s_0} \right).
$

Since Lemma \ref{lem:total_bound_matroid_partition} implies that $\displaystyle s_0 = O\left( \frac{p}{d} \right)$,
the number of calls of \texttt{EdgeRecyclingAugmentation} in Line \ref{line:faster_matorid_partition_2} of Algorithm \ref{alg:FasterMatroidPartitionForLaergek} is $\displaystyle O\left(\sqrt{\frac{p}{d}}\right)$.
By Lemma \ref{lem:faster_matroid_for_k_inequality_1}, the proof is complete.
\end{proof}


In Algorithm \ref{alg:FasterMatroidPartitionForLaergek},
we set a parameter $d$ in order to balance the number of independence oracle queries used in Lines \ref{line:faster_matorid_partition_1} and \ref{line:faster_matorid_partition_2}.
Thus we obtain the following proof.

\begin{proof}[Proof of Theorem~\ref{thm:faster_matroid_partition_ind_oracle_2}]
Let $k' = \min \{k, p\}$.
By Lemma \ref{lem:analysis_matroid_parition_for_large_k_1},
we observe that the number of independence oracle queries used in Line \ref{line:faster_matorid_partition_1} of Algorithm \ref{alg:FasterMatroidPartitionForLaergek} is bounded by $O(k' n d \log p + k n)$.
Here, we set $\displaystyle d = \frac{\bar{p}}{k'^{2/3}}$ and run Algorithm \ref{alg:FasterMatroidPartitionForLaergek}.
Then, by Lemmas \ref{lem:analysis_matroid_parition_for_large_k_1} and \ref{lem:analysis_matroid_parition_for_large_k_2}, 
the number of independence oracle queries used in Lines \ref{line:faster_matorid_partition_1} and \ref{line:faster_matorid_partition_2} is $O(k'^{1/3} n p \log p + k n)$.
Furthermore, the number of independence oracle queries used in Lines \ref{line:faster_matorid_partition_4} and \ref{line:faster_matorid_partition_3} is $O(k n)$, which completes the proof.
\end{proof}

Note that Algorithm \ref{alg:FasterMatroidPartitionForLaergek} requires \tatsuya{$\displaystyle O\left(n p \cdot \frac{p}{d}\right) = O(k'^{2/3} n p)$} time complexity other than independence oracle queries.
This is because we use the edge set $E^*$ of size $n p$ in \texttt{EdgeRecyclingBFS} and the number of total \texttt{EdgeRecyclingBFS} calls in Algorithm \ref{alg:FasterMatroidPartitionForLaergek} is $\displaystyle O\left( \frac{p}{d}\right)$.

\section{Concluding Remarks}\label{sec:conclusion}
By simply combining Cunningham's algorithm~\cite{cunningham1986improved} and the binary search technique proposed by
\cite{nguyen2019note, chakrabarty2019faster},
we can not break the $O(n^{5/2})$-independence-query bound for the matroid partitioning problem.
However, we introduce a new approach \emph{edge recycling augmentation} and break this barrier and obtain an algorithm that uses $\tO(n^{7/3})$ independence oracle queries.
This result will be a substantial step forward understanding the matroid partitioning problem.

Our main observation is that some edges of the compressed exchange graph will remain the same after an augmentation, and then we need not query again to find them.
This yields a matroid partition algorithm whose independence query complexity is sublinear in $k$.
This idea is quite simple,
and we believe that our new approach edge recycling augmentation will be useful in the design of future algorithms.

In a recent breakthrough,
Blikstad-Mukhopadhyay-Nanongkai-Tu~\cite{blikstad2021breaking_STOC} broke the $\tO(n^2)$-independence-query bound for matroid intersection.
Then it is natural to ask whether we can make a similar improvement for the matroid partition algorithm.
However, such an improvement is impossible.
As one anonymous reviewer of ICALP 2023 pointed out, it is easy to show that the matroid partitioning problem requires $\Omega(k n)$ independence oracle queries, which is $\Omega(n^2)$ when $k = \Theta(n)$.\footnote{
    Let $M_1, \ldots, M_k$ be matroids of rank $1$ defined over a common ground set $V$ of $n$ elements.
    Now, we construct a bipartite graph $G = (L \cup R, E)$ with $|L| = n$, $|R| = k$ where $(v, i) \in E$ if and only if $\{v\}$ is independent in $M_i$.
    Here, the maximum size of a partitionable set is equal to the maximum size of a matching in $G$.
    It can be viewed as having edge-query access to this graph $G$, since it does not make sense to use an independence query to a set of size at least $2$.
    It is well-known that it needs $\Omega(k n)$ edge queries to find the size of a maximum matching; see \cite{yao1988monotone} for details.
}
Then, there is a clear difference between these two problems.



We also consider a matroid partition algorithm in the rank oracle model and present a matroid partition algorithm that uses $\tO(n^{3/2})$ rank oracle queries when $k \leq n$.
Blikstad et al.~\cite{blikstad2021breaking_STOC} asks whether the tight bounds of the matroid intersection problem are the same under independence oracle model and rank oracle model.
The same kind of problem is natural for the matroid partitioning problem.
Unlike the matroid intersection problem,
we believe there exists a difference between independence oracle and rank oracle in terms of query complexity of the matroid partitioning problem.



\section*{Acknowledgements}
The author thanks Yusuke Kobayashi for his generous support and helpful comments on the manuscript.
The author also thanks the three anonymous reviewers of ICALP 2023 for their valuable comments.
This work was partially supported by the joint project of Kyoto University and Toyota Motor Corporation, titled ``Advanced Mathematical Science for Mobility Society''.

\bibliography{biblio}

\end{document}